\documentclass{IEEEojcsys}

\usepackage[utf8]{inputenc}
\usepackage{amsmath}
\usepackage{amssymb}
\usepackage{graphicx,balance}
\usepackage{pifont}
\usepackage{bm}
\usepackage{dsfont}
\usepackage{subcaption}
\usepackage{tikz}
\usepackage{booktabs}

\usepackage{xcolor}

\definecolor{t0}{rgb}{1.0, 0.6470588235294115, 0.0}
\definecolor{t1}{rgb}{0.7584702998846602, 0.9377450980392157, 0.0}
\definecolor{t2}{rgb}{0.2317474048442904, 0.8754901960784314, 0.0}
\definecolor{t3}{rgb}{0.0, 0.8132352941176471, 0.22722750865051902}
\definecolor{t4}{rgb}{0.0, 0.7509803921568627, 0.6184544405997694}
\definecolor{t5}{rgb}{0.0, 0.43551758938869645, 0.6887254901960784}
\definecolor{t6}{rgb}{0.0, 0.055276816608996346, 0.6264705882352941}
\definecolor{t7}{rgb}{0.257215974625144, 0.0, 0.5642156862745098}
\definecolor{t8}{rgb}{0.5019607843137255, 0.0, 0.5019607843137254}

\def\yshift{0.7}
\def\xshift{-7.5}
\definecolor{t02}{HTML}{FFA500}
\definecolor{t12}{HTML}{8EE345}
\definecolor{t22}{HTML}{5CC95E}
\definecolor{t32}{HTML}{3A83AF}
\definecolor{t42}{HTML}{110193}
\definecolor{t52}{HTML}{75147C}

\usepackage{amsfonts}
\usepackage{url}
\usepackage[pdftex,plainpages = false, colorlinks=true, linkcolor=black, citecolor = black, urlcolor = blue,pagebackref=false,hypertexnames=false, plainpages=false, pdfpagelabels]{hyperref}
\usepackage{balance}
\usepackage[capitalize]{cleveref}
\usepackage[style=ieee, url=false, doi=false, natbib=true, mincitenames=1, maxcitenames=1, maxbibnames=99]{biblatex}
\usepackage{svg}

\newtheorem{theorem}{Theorem}[section]
\newtheorem{corollary}[theorem]{Corollary}
\newtheorem{lemma}[theorem]{Lemma}

\crefname{section}{Section}{Sections}
\crefname{theorem}{Theorem}{Theorems}
\crefname{lemma}{Lemma}{Lemmas}
\crefname{table}{Table}{Tables}
\crefformat{equation}{(#2#1#3)}
\crefname{algocf}{Algorithm}{Algorithms}
\Crefname{algocf}{Algorithm}{Algorithms}
\crefname{ALC@unique}{Line}{Lines}

\newcommand{\x}{\mathbf{x}}

\newcommand{\R}{\mathds{R}}

\newcommand{\istate}{k}

\newcommand{\xnom}{\mathring{\mathbf{x}}}

\newcommand{\bpxnom}{\mathcal{B}_p(\xnom, \bm{\epsilon)}}

\newcommand{\CROWNAu}{\mathbf{\Psi}}
\newcommand{\CROWNAl}{\mathbf{\Phi}}
\newcommand{\CROWNbu}{\bm{\alpha}}
\newcommand{\CROWNbl}{\bm{\beta}}

\usepackage{accents}
\newcommand{\ubar}[1]{\underaccent{\bar}{#1}}

\newcommand{\nstep}{ReBReach-LP}
\newcommand{\basic}{BReach-LP}
\newcommand{\hybrid}{HyBReach-LP}
\newcommand{\BRNL}{BReach-MILP}
\usepackage{paralist}

\usepackage[addedmarkup=bf]{changes}
\definechangesauthor[name={MFE}, color={blue}]{MFE}
\definechangesauthor[name={JH}, color={red}]{JH}
\definecolor{purple}{RGB}{136, 0, 199}
\definechangesauthor[name={NR}, color={purple}]{NR}
\definecolor{cardinal}{RGB}{177, 4, 14}
\definechangesauthor[name={Stanford}, color={cardinal}]{Stanford}
\crefformat{chapter}{\S#2#1#3}
\crefmultiformat{chapter}{\S\S#2#1#3}{and~#2#1#3}{, #2#1#3}{, and~#2#1#3}
\crefformat{section}{\S#2#1#3}
\crefmultiformat{section}{\S\S#2#1#3}{and~#2#1#3}{, #2#1#3}{, and~#2#1#3}

\usepackage{algorithm,algorithmic}

\usepackage{outlines}
\usepackage{colortbl}

\usepackage{thmtools}
\usepackage{thm-restate}

\jvol{00}
\jnum{XX}
\paper{1234567}
\pubyear{2021}
\receiveddate{XX September 2021}
\accepteddate{XX October 2021}
\publisheddate{XX November 2021}
\currentdate{XX November 2021}
\doiinfo{OJCSYS.2021.Doi Number}

\begin{document}

\sptitle{Article Category}

\newcommand{\thetitle}{Backward Reachability Analysis of Neural Feedback Loops: Techniques for Linear and Nonlinear Systems}
\title{\thetitle}

\editor{This paper was recommended by Associate Editor F. A. Author.}

\author{Nicholas Rober\affilmark{1} (Student Member, IEEE)}

\author{Sydney M. Katz\affilmark{2} (Student Member, IEEE)}

\author{Chelsea Sidrane\affilmark{2}}

\author{Esen Yel\affilmark{2} (Member, IEEE)} 

\author{Michael Everett\affilmark{1} (Member, IEEE)}

\author{Mykel J. Kochenderfer\affilmark{2} (Senior Member, IEEE)}

\author{Jonathan P. How\affilmark{1} (Fellow, IEEE)}

\affil{Department of Aeronautics \& Astronautics, Massachusetts Institute of Technology, Cambridge, MA 02139 USA} 
\affil{Department of Aeronautics \& Astronautics, Stanford University, Stanford, CA 94305 USA}

\corresp{CORRESPONDING AUTHOR: Nicholas Rober (e-mail: \href{mailto:nrober@mit.edu}{nrober@mit.edu})}
\authornote{This work was supported in part by Ford Motor Company. The NASA University Leadership Initiative (grant \#80NSSC20M0163) also provided funds to assist the authors with their research. This research was also supported by the National Science Foundation Graduate Research Fellowship under Grant No. DGE–1656518. Any opinion, findings, and conclusions or recommendations expressed in this material are those of the authors and do not necessarily reflect the views of any NASA entity or the National Science Foundation. This work was also supported by AFRL and DARPA under contract FA8750-18-C-0099.}

\markboth{\MakeUppercase{\thetitle}}{\MakeUppercase{Nicholas Rober} {\itshape ET AL}.}

% \author{
% \thanks{Aerospace Controls Laboratory, Massachusetts Institute of Technology, Cambridge, USA. e-mail: {\tt \small \{nrober,mfe,jhow\}@mit.edu}. Research supported by Ford Motor Company. \stan{Add research support}}}

\begin{abstract}
{\color{black}
As neural networks (NNs) become more prevalent in safety-critical applications such as control of vehicles, there is a growing need to certify that systems with NN components are safe.
This paper presents a set of backward reachability approaches for safety certification of neural feedback loops (NFLs), i.e., closed-loop systems with NN control policies.
While backward reachability strategies have been developed for systems without NN components, the nonlinearities in NN activation functions and general noninvertibility of NN weight matrices make backward reachability for NFLs a challenging problem.
To avoid the difficulties associated with propagating sets backward through NNs, we introduce a framework that leverages standard forward NN analysis tools to efficiently find over-approximations to backprojection (BP) sets, i.e., sets of states for which an NN policy will lead a system to a given target set.
We present frameworks for calculating BP over-approximations for both linear %\footnote{\textbf{Code}: \url{https://github.com/mit-acl/nn_robustness_analysis}} 
and nonlinear %\footnote{\textbf{Code}: \url{https://github.com/smkatz12/BackwardsReach.git}}
systems with control policies represented by feedforward NNs and propose computationally efficient strategies.
We use numerical results from a variety of models to showcase the proposed algorithms, including a demonstration of safety certification for a 6D system.
}

\end{abstract}

\begin{IEEEkeywords}
Reachability, Neural Networks, Safety, Neural Feedback Loops, Linear Systems, Nonlinear Systems
% Enter key words or phrases in alphabetical order, separated by commas. For a list of\goodbreak suggested keywords, send a blank e-mail to \href{mailto:keywords@ieee.org}{keywords@ieee.org} or visit\goodbreak \href{http://www.ieee.org/organizations/pubs/ani_prod/keywrd98.txt}{http://www.ieee.org/organizations/pubs/ani\_prod/keywrd98.txt}
\end{IEEEkeywords}

\maketitle

%!TEX root=main.tex

\section{Introduction}
% Neural networks (NNs) play an important role in many modern robotic systems.
Neural network-based control has seen success in domains that remain difficult for more classical control methods, such as robotic locomotion over varied terrain~\cite{tan2018sim, lee2020learning}, recovery of a quadrotor from aggressive maneuvers~\cite{hwangbo2017control}, and robotic manipulation~\cite{gu2017deep, kalashnikov2018scalable}. 
Neural network-based control also offers complementary properties compared to classical linear control and more recent optimal control approaches~\cite{franklin2014feedback, grune2017nonlinear}.
Neural network-based control is capable of handling complex nonlinear systems, and like linear control it can be fast to execute.
However, the standard training procedures for NN-based control do not provide the convergence and safety guarantees that can usually be derived for more classical control approaches. 
Further, despite achieving high performance during testing, many works have demonstrated that NNs can be sensitive to small perturbations in the input space and produce unexpected behavior or incorrect conclusions as a result~\cite{kurakin2016adversarial, yuan2019adversarial}.
Thus, before applying NNs to safety-critical systems such as self-driving cars~\cite{chen2015deepdriving, kendall2019learning, kelly2019hg} and aircraft collision avoidance systems~\cite{julian2019deep}, there is a need for tools that provide safety guarantees.
Providing safety guarantees presents computational challenges due to the high dimensional and nonlinear nature of NNs.

\begin{figure}[t]
\setlength\belowcaptionskip{-0.7\baselineskip}
\centering
\captionsetup[subfigure]{aboveskip=-1pt,belowskip=-1pt}
    \begin{subfigure}[t]{\columnwidth}
        \includegraphics[width=\columnwidth]{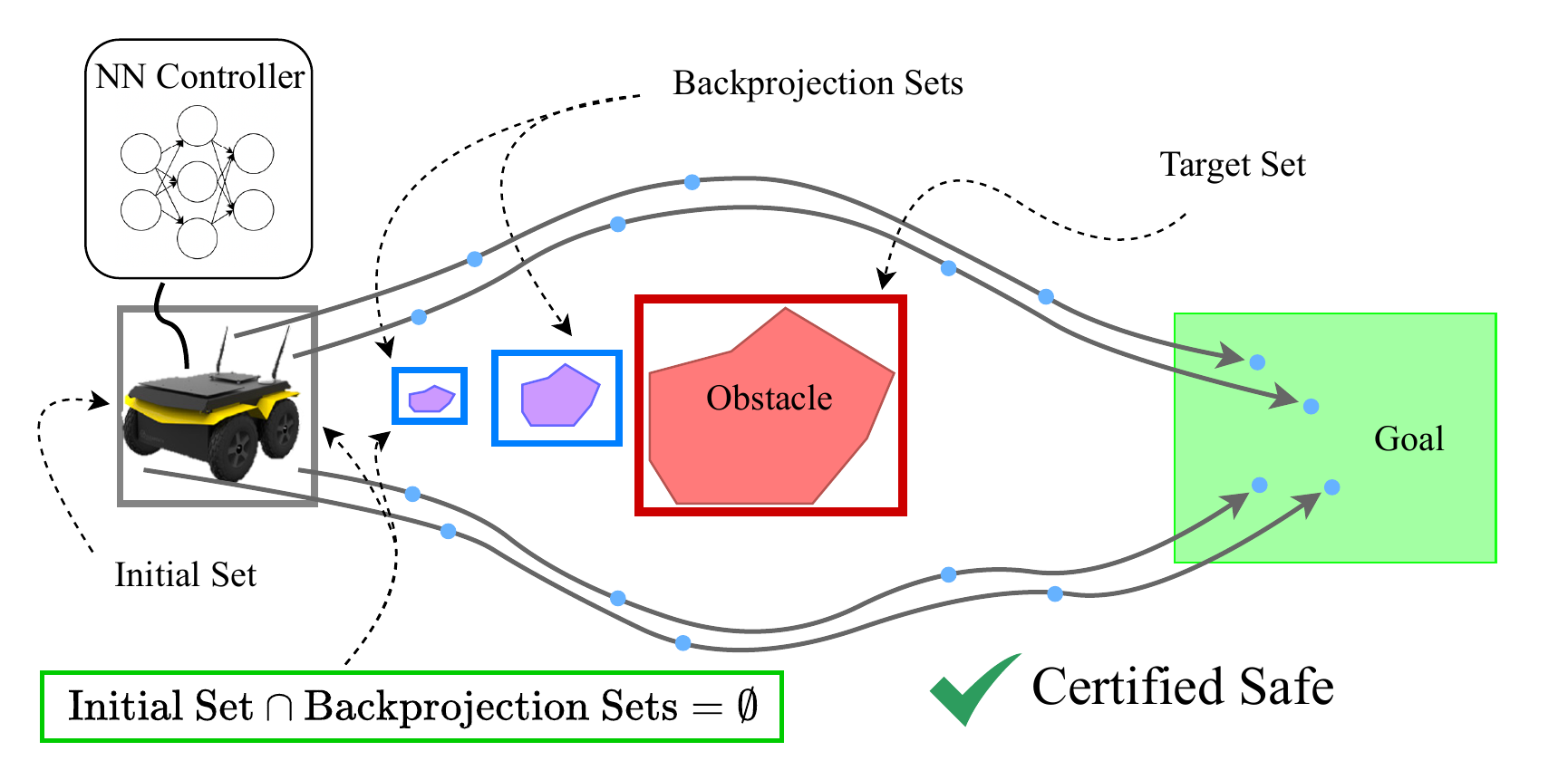}
        \caption{{\color{black}Backward reachability for collision avoidance: BP set estimates (blue) approximate the set of states that lead to the obstacle (red), thus if the initial state set does not intersect with any BPs, the situation is safe.}}
        \label{fig:reach_comp:backward}
    \end{subfigure}
    \begin{subfigure}[t]{\columnwidth}
        \includegraphics[width=\columnwidth]{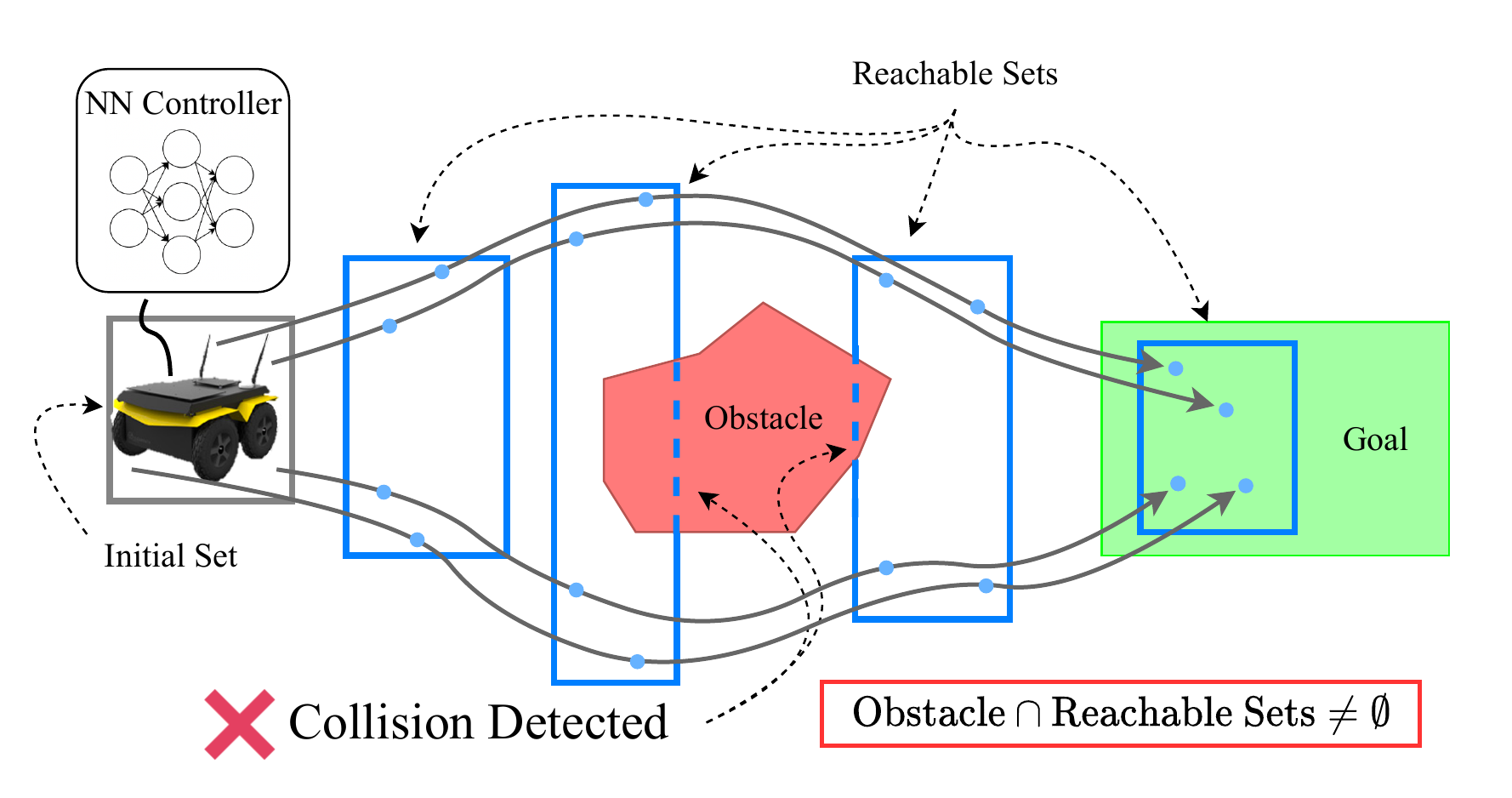}
        \caption{{\color{black}Forward reachability for collision avoidance: forward reachable set estimates (blue) approximate the set of possible future states of the system, thus any intersection with an obstacle means safety cannot be certified.}}
        \label{fig:reach_comp:forward}
    \end{subfigure}
    \caption{{\color{black}Collision avoidance scenario. Backward reachability correctly guarantees safety whereas forward reachability fails.}}
    \label{fig:reach_comp}
\end{figure}

There is a growing body of work focused on synthesizing NN controllers with safety and performance guarantees~\cite{chang2019neural, sun2020learning, han2020actor, qin2021learning, dai2021lyapunov, dawson2022safe}, but this does not preclude the need for verification and safety analysis after synthesis.
{\color{black}To this end, many open-loop NN analysis tools \cite{zhang2018efficient, weng2018towards, xu2020automatic, raghunathan2018semidefinite, tjeng2017evaluating, katz2019marabou, katz2017reluplex, jia2021verifying,vincent2021reachable} have been developed to make statements about possible NN outputs given a set of inputs.
To extend analysis techniques to closed-loop systems with NN controllers, i.e., neural feedback loops (NFLs), there has also been effort towards developing reachability analysis techniques \cite{vincent2021reachable,dutta2019reachability, huang2019reachnn, ivanov2019verisig, fan2020reachnn, xiang2020reachable, hu2020reach, sidrane2021overt, everett2021reachability,chen2022one,  bak2022closed} that determine how a system's state evolves over time.
While forward reachability \cite{vincent2021reachable,dutta2019reachability, huang2019reachnn, ivanov2019verisig, fan2020reachnn, xiang2020reachable, hu2020reach, sidrane2021overt, everett2021reachability, chen2022one} certifies safety by checking that possible future states of the system do not enter dangerous regions, this paper focuses on backward reachability \cite{bak2022closed}, wherein safety is certified by checking that the system does not start from a state that could lead to a dangerous region, as shown in \cref{fig:reach_comp:backward}.
Thus, the challenge of backward reachability analysis is to calculate \textit{backprojection} (BP) \textit{sets} that define regions of the state space for which the NN control policy will drive the system into a given \textit{target set}, which can be chosen to contain a dangerous part of the state space.

Contrasting \cref{fig:reach_comp:backward} with \cref{fig:reach_comp:forward} shows how backward reachability can be less conservative than forward reachability in scenarios where there are multiple modes for possible trajectories branching from a given initial state set.
In \cref{fig:reach_comp:forward}, the robot's position within the initial state set determines whether it will go above or below the obstacle.
The forward reachable sets must contain all possible future trajectories, so when the forward reachable sets are represented by single convex sets (as is the case in \cite{dutta2019reachability, huang2019reachnn, ivanov2019verisig, fan2020reachnn, xiang2020reachable, hu2020reach, sidrane2021overt, everett2021reachability}), they must span the upper and lower trajectories, thus detecting a possible collision with the obstacle and failing to certify safety.
Alternatively, because the BP sets in \cref{fig:reach_comp:backward} do not intersect with the initial state set, the robot is not among the states that leads to an obstacle, allowing backward reachability to certify safety.

For systems without NNs, switching between forward and backward reachability analysis is relatively simple \cite{bansal2017hamilton,evans2010partial,mitchell2007comparing}.
However, nonlinear activation functions and noninvertible weight matrices common to NNs make it challenging to determine a set of inputs given a set of outputs.
While recent work \cite{vincent2021reachable,bak2022closed,everett2021reachability} has made advances in backward reachability of NFLs, there are no existing techniques that efficiently find BP set estimates over multiple timesteps for linear and nonlinear systems with feedforward NNs that give continuous outputs.
}
Thus, our contributions include:
\begin{itemize}
    \item A set of algorithms that enable computationally efficient safety certification of linear and nonlinear NFLs by calculating over-approximations of BP sets.
    \item Validation of our methods using numerical experiments for control-relevant applications including obstacle avoidance for mobile robots and quadrotors.
\end{itemize}
This work extends prior work \cite{rober2022backward} by introducing:
\begin{compactitem}
    \item \hybrid{}: A hybrid of the two previously proposed algorithms \cite{rober2022backward} that improves the tradeoff between conservativeness and computation time.
    \item A guided partitioning algorithm that reduces conservativeness faster than uniform partitioning strategies.
    \item \BRNL{}: an algorithm to compute BP sets for systems with nonlinear dynamics using techniques developed by \citet{sidrane2021overt}.
    \item Numerical experiments that exhibit our BP estimation techniques on higher-order and nonlinear systems, including an ablation study marking improvements from  \cite{rober2022backward}.
\end{compactitem}

\section{Related Work}
{\color{black}This section describes how reachability analysis can be applied to three categories of systems: NNs in isolation (i.e., open-loop analysis), closed-loop systems without neural components, and NFLs.

Open-loop NN analysis refers to methods that determine a relation between sets of inputs to an NN and sets of the NN's output.
Open-loop NN analysis encompasses techniques that relax the nonlinearities in the NN activation functions to quickly provide relatively conservative bounds on NN outputs~\cite{zhang2018efficient, raghunathan2018semidefinite}, and techniques that take more time to provide exact bounds~\cite{katz2017reluplex, vincent2021reachable}.
Often, these tools are motivated by the threat posed by adversarial attacks \cite{raghunathan2018semidefinite, tjeng2017evaluating} that can drastically affect the output of perception models.
To extend the use of open-loop analysis tools to NFLs, the closed-loop system dynamics must be considered.

Reachability analysis is a well-established method for safety verification of closed-loop systems without NN components.
Hamilton-Jacobi methods \cite{bansal2017hamilton,evans2010partial} provide the main theoretical framework for these analyses, and computational tools such as CORA~\cite{althoff2015introduction}, Flow*~\cite{chen2013flow}, SpaceEx~\cite{frehse2011spaceex}, HYLAA~\cite{bak2017hylaa}, and C2E2~\cite{duggirala2015c2e2} can be used to compute reachable sets for a given system.
However, these tools cannot currently be used to analyze NFLs.

When it comes to NFLs, forward reachability analysis is the focus of  recent work~\cite{vincent2021reachable,dutta2019reachability, huang2019reachnn, ivanov2019verisig, fan2020reachnn, xiang2020reachable, hu2020reach, sidrane2021overt, everett2021reachability, chen2022one}.
Unlike traditional reachability analysis techniques such as Hamilton-Jacobi methods \cite{bansal2017hamilton,evans2010partial}, backward reachability analysis of NFLs introduces significant challenges that are not present in forward reachability of NFLs.
Forward reachability of NFLs relies on open-loop NN analysis tools that are configured to pass sets forward through an NN.
This suggests that backward reachability would need a way to propagate sets backward through an NN, but this presents two major challenges:
First, many activation functions do not have a one-to-one mapping of inputs to outputs (consider $\mathrm{ReLU}(x) = 0$, which corresponds to any value $x \leq 0$).
In these cases, propagating sets backward through an NN can cause an infinitely large BP sets.
Additionally, the weight matrices in each layer of the NN may be singular or rank deficient, meaning they are not generally invertible, which presents another fundamental challenge in determining a set of inputs from a set of outputs.
}

While recent works on NN inversion have developed NN architectures that are designed to be invertible~\cite{ardizzone2018analyzing} and training procedures that regularize for invertibility~\cite{behrmann2019invertible}, dependence on these techniques would be a major limitation on the class of systems for which backward reachability analysis could be applied.
Our approach avoids the challenges associated with NN-invertibility.
For linear systems, our approach can be applied to the same class of NN architectures as CROWN \cite{zhang2018efficient}, i.e., NNs for which an affine relaxation can be found.
For nonlinear systems, our approach can be applied to any NN with piecewise-linear activation functions (ReLU, Leaky ReLU, hardtanh, etc.). 

% While the body of work concerning reachability analysis is growing, less attention has been paid to the reverse problem of finding \emph{backprojection} sets, i.e. the set of possible states for which the NN control policy will deliver the system to a target set. 
{\color{black}
Recent work investigated backwards reachability analysis for NFLs.
\Citet{vincent2021reachable} describe a method for backward reachability on an NN dynamics model, but this method is limited to NNs with ReLU activations.
Alternatively, \citet{bak2022closed} use a quantized state approach~\cite{jia2021verifying} to quickly compute BP sets, but it requires an alteration of the original NN through a preprocessing step that can affect its overall behavior and does not consider continuous NN outputs.
Finally, \citet{everett2021reachability} derive a closed-form equation that can be used to find BP set under-approximations, but under-approximations cannot be used to guarantee safety because they may miss some states that reach the target set (under-approximations are most useful for goal checking, where it is desirable to guarantee that all states in the BP estimate reach the target set).
This work expands upon the existing literature on backward reachability for NFLs by introducing a framework for efficiently calculating BP set over-approximations for NFLs with continuous action spaces.
}

\section{Preliminaries}

\subsection{System Dynamics}
We consider a discrete-time time-invariant system with dynamics given by
\begin{equation} \label{eq:general_dyn}
    \mathbf{x}_{t+1} = p(\mathbf{x}_{t}, \mathbf{u}_t),
\end{equation}
where $\mathbf{x}_t \in \mathds{R}^{n_x}$ is the state, $n_x$ is the state dimension, $\mathbf{u}_t \in \mathds{R}^{n_u}$ is the control input, and $n_u$ is the control dimension.
In the linear case, \cref{eq:general_dyn} may be written
\begin{equation} \label{eqn:lti_dynamics}
    \mathbf{x}_{t+1} = \mathbf{A}\mathbf{x}_t + \mathbf{B}\mathbf{u}_t +\mathbf{c}.
\end{equation}
The matrices $\mathbf{A}$, $\mathbf{B}$, and $\mathbf{c}$ are assumed known.
In both linear and nonlinear cases, the control input $\mathbf{u}_t$ is generated by a state-feedback control policy specified by an $m$-layer feedforward NN  $\pi(\cdot)$, i.e., $\mathbf{u}_t = \pi(\mathbf{x}_t)$.
The closed-loop system \cref{eq:general_dyn} with control policy $\pi$ is denoted
\begin{equation}
    \mathbf{x}_{t+1} = f(\mathbf{x}_{t}; \pi). \label{eqn:closed_loop_dynamics}
\end{equation}
We assume $\mathbf{u}_t$ is constrained by control limits and is thus limited to a convex set of allowable control values $\mathcal{U}$, i.e., $\mathbf{u}_t \in \mathcal{U}$. 
Similarly, we consider $\mathbf{x}_t$ to be constrained to a convex operating region of the state space $\mathcal{X}$, i.e., $\mathbf{x}_t \in \mathcal{X}$.
% Note that for systems that can be written as $\mathbf{x}_{t+1} = f(\mathbf{u}_t)$, 
Note that for nonlinear systems, $\mathcal{X}$ represents the set of states that are physically realizable for a given system, e.g., the maximum speed of a vehicle that cannot be exceeded regardless of input. 
We implement this by using clipping in the dynamics. 
However, for linear systems, we do not clip, as this would introduce non-linearity.
%$\mathcal{X}$ should be considered an operating region that the system is assumed to stay within.
Defining $\mathcal X$ is particularly important for backward reachability, because the target set says nothing about initial states that are practical.
For example, a very high velocity could cause the positional range of the backprojection set to also be quite large, but the practitioner might know that the system would never start from a set of states that allow for such a high velocity to be achieved in the time horizon. 
To encode this idea in the optimization problems, we assume $\mathcal{X}$ is given and contains all possible previous states.

\subsection{Nonlinear System Over-Approximation}
We calculate backreachable and backprojection sets for nonlinear systems by extending the mixed integer linear programming (MILP) methods of \citet{sidrane2021overt}, which over-approximate nonlinear dynamics using piecewise-linear representations.
Using a piecewise-linear over-approximation of the nonlinear dynamics preserves the ability to calculate global optima for the optimization problem as well as make sound claims about the original nonlinear system.
As described by \citet{sidrane2021overt}, the closed-loop system dynamics function \cref{eq:general_dyn} mapping $\mathbb{R}^{n_x} \rightarrow \mathbb{R}^{n_x}$ is considered as $n_x$ functions mapping $\mathbb{R}^{n_x} \rightarrow \mathbb{R}$. Each function mapping $\mathbb{R}^{n_x} \rightarrow \mathbb{R}$ is re-written into a series of affine functions and $\mathbb{R} \rightarrow \mathbb{R}$ nonlinear functions.
For each nonlinear function $e(x)$ defined over the interval $x \in [a,b]$, an optimally-tight piecewise-linear upper and lower bound is computed such that:
$g_{e_{LB}}(x) \leq e(x) \leq g_{e_{UB}}(x)~~\forall x \in [a,b]$.
Together, the affine functions and piecewise-linear lower and upper bounds form an abstraction $\hat{p}$ of the original nonlinear closed-loop system $p$~\cref{eq:general_dyn}.
% The use of lower and upper bounds mean that $\x_{t+1}$ no longer has a \textit{functional} relationship to $\x_t$ and $\u_t$.
With control policy $\mathbf{u}_t = \pi(\x_t)$, we denote the abstract closed-loop system relating $\x_t$ and $\x_{t+1}$ as $\hat{f}(\x_t;\pi)$.
The abstract closed-loop system represents an over-approximation of $f(\cdot)$, or more formally $\forall \x_t \in \mathcal{X}\nonumber$,
% \begin{align}
% &\forall \x_t \in \mathcal{X}\nonumber \\
% &\{\x_{t+1} \lvert\ \x_{t+1} = f(\x_t;\pi) \} \subseteq \{\x_{t+1} \lvert\ \x_{t+1} \bowtie \hat{f}(\x_t; \pi)\}
% \end{align}
\begin{equation}
    \{\x_{t+1} \lvert\ \x_{t+1} = f(\x_t;\pi) \} \subseteq \{\x_{t+1} \lvert\ \x_{t+1} \bowtie \hat{f}(\x_t; \pi)\}
\end{equation}
where we denote that a state $\x_{t+1}$ is a valid successor of state $\x_t$ under the abstract system if $\x_{t+1} \bowtie \hat{f}(\x_t; \pi)$.

% To calculate backreachable and backprojection sets for nonlinear systems, this paper builds upon the approaches presented in \cite{sidrane2021overt, rober2022backward} which rely on solving LPs and Mixed-Integer Linear Programs (MILPs). 
% LPs and MILPs are well-studied problems for which one can reliably find the global optimum.

\subsection{Control Policy Neural Network Structure}

% For a feedforward NN with $L$ hidden layers (and two additional input and output layers), the number of neurons in each layer is $n_l\in[L+1]$, where $[i]$ denotes the set $\{0,1,\ldots,i\}$.
% We introduce the $l$-th layer's weight matrix $\mathbf{W}^{(l)}\in\R^{n_{l+1}\times n_{l}}$, bias vector $\mathbf{b}^{(l)}\in\R^{n_{l+1}}$, and coordinate-wise activation function $\sigma^{(l)}: \R^{n_{l+1}} \to \R^{n_{l+1}}$, which can include tanh, sigmoid, ReLU (i.e., $\sigma(\mathbf{z})=\mathrm{max}(0,\mathbf{z})$), and many others.
% For input $\x\in\R^{n_0}$, the NN output $\pi(\x)$ is,
% \begin{align}
% \begin{split}
%     \x^{(0)} &= \x \\
%     \z^{(l)} &= \mathbf{W}^{(l)} \x^{(l)}+\mathbf{b}^{(l)}, \forall l\in[L] \\
%     \x^{(l+1)} &= \sigma^{(l)}(\z^{(l)}), \forall l\in[L-1] \\
%     \pi(\x) &= \z^{(L)}.
% \end{split}
% \end{align}

Consider a feedforward NN with $M$ hidden layers, and one input and one output layer.
We denote the number of neurons in each layer as $n_m,\ \forall m \in [M+1]$ where  $[i]$ denotes the set $\{0,1,\ldots,i\}$.
The $m$-th layer has weight matrix $\mathbf{W}^{m}\in\R^{n_{m+1}\times n_{m}}$, bias vector $\mathbf{b}^{m}\in\R^{n_{m+1}}$, and activation function $\sigma^{m}: \R^{n_{m+1}} \to \R^{n_{m+1}}$.
For linear systems, $\sigma^{m}$ can be any option handled by CROWN~\cite{zhang2018efficient}, e.g., sigmoid, tanh, ReLU, and for nonlinear systems $\sigma^{m}$ can be any piecewise-linear function, e.g. ReLU, Leaky ReLU, hardtanh, etc. 
For an input $\x\in\R^{n_0}$, the NN output $\pi(\x)$ is computed as
% \begin{align}
% \begin{split}
%     \x^{0} &= \x \\
%     \z^{l} &= \mathbf{W}^{l} \x^{l}+\mathbf{b}^{l}, \forall l\in[L] \\
%     \x^{l+1} &= \sigma^{l}(\z^{l}), \forall l\in[L-1] \\
%     \pi(\x) &= \z^{L}.
% \end{split}
% \end{align}
\begin{align}
\begin{split}
    \x^{0} &= \x \\
    \x^{m+1} &= \sigma^{m}(\mathbf{W}^{m} \x^{m}+\mathbf{b}^{m}), \forall m\in[M-1] \\
    \pi(\x) &= \mathbf{W}^{M} \x^{M}+\mathbf{b}^{M}.
\end{split}
\end{align}

\subsection{Neural Network Relaxation}

% A key step in efficiently conducting reachable set analysis of an NFL~\cref{eqn:closed_loop_dynamics} is to relax the nonlinearities introduced by the NN's activation functions. By converting the nonlinearities in the activation function into an affine upper and lower bound, where each bound holds for the known range of inputs to the NN, we can obtain a simplified relationship between the NN's input and output. 

% We represent the range of inputs to the activation functions with the $\epsilon$-ball (also called the $\ell_p$-ball).
% Denote the $\epsilon$-ball under the $\ell_p$ norm, centered at $\xnom$, with scalar radius, $\epsilon$,
% \begin{align}
%     \bpxnomscalar &= \{\mathbf{x}\ \lvert\ \lvert\lvert \mathbf{x} - \xnom \rvert \rvert_p \leq \epsilon \}.
% \end{align}
% The $\epsilon$-ball is extended to the case of vector $\bm{\epsilon}\in\R^n_{\geq 0}$ (i.e., $\bm{\epsilon}$-ball), defined as
% \begin{align}
%     \bpxnom &= \{\mathbf{x}\ \lvert\ \lim_{\bm{\epsilon}' \to \bm{\epsilon}^+} \lvert\lvert (\mathbf{x} - \xnom) \oslash \bm{\epsilon}' \rvert\rvert_p \leq 1\},
% \end{align}
% where $\oslash$ denotes element-wise division.

%To avoid the computational cost associated with calculating exact BP sets, 
For linear dynamics,
we relax the NN's activation functions to obtain affine bounds on the NN outputs for a known set of inputs.
This preserves the problem as an LP and reduces computational cost.
The range of inputs are represented using the $\ell_p$-ball 
% \begin{align}
%     \bpxnom &\triangleq \{\mathbf{x}\ \lvert\ \lim_{\bm{\epsilon}' \to \bm{\epsilon}^+} \| (\mathbf{x} - \xnom) \oslash \bm{\epsilon}' \|_p \leq 1\},
% \end{align}
\begin{align}
    \bpxnom &\triangleq \{\mathbf{x}\ \lvert\ \| (\mathbf{x} - \xnom) \oslash \bm{\epsilon} \|_p \leq 1\},
\end{align}
where $\xnom \in \R^n$ is the center of the ball, $\bm{\epsilon}\in\R^n_{\geq 0}$ is a vector whose elements are the radii for the corresponding elements of $\mathbf{x}$, and $\oslash$ denotes element-wise division.

% Given an $m$-layer neural network function $f : \R^{n_0} \rightarrow \R^{n_m}$, there exists two explicit functions $f^L_j : \R^{n_0} \rightarrow \R$ and $f^U_j :\R^{n_0} \rightarrow \R$ such that $\forall j \in [n_m], \; \forall \x \in \Ball$, the inequality $\, f_{j}^{L}(\x) \leq f_{j}(\x) \leq f_{j}^{U}(\x)$ holds true, where
\begin{theorem}[\!\!\cite{zhang2018efficient}, Convex Relaxation of NN]\label{thm:crown_particular_x}
Given an $m$-layer neural network control policy $\pi:\R^{n_x}\to\R^{n_u}$, there exist two explicit functions $\pi_j^L: \R^{n_x}\to\R^{n_u}$ and $\pi_j^U: \R^{n_x}\to\R^{n_u}$ such that $\forall j\in [n_m], \forall \mathbf{x}\in\bpxnom$, the inequality $\pi_j^L(\mathbf{x})\leq \pi_j(\mathbf{x})\leq \pi_j^U(\mathbf{x})$ holds true, where
\begin{equation}
\label{eq:f_j_UL}
    \pi_{j}^{U}(\x) = \CROWNAu_{j,:} \x + \CROWNbu_j, \quad
    \pi_{j}^{L}(\x) = \CROWNAl_{j,:} \x + \CROWNbl_j,
    % \pi_{j}^{U}(\y) &= \Au{(0)}_{j,:} \y + \sum_{k=1}^{m}\Au{(k)}_{j,:}(\bias{(k)}+\upbias{(k)}_{:,j}) \\
    % \pi_{j}^{L}(\y) &= \Al{(0)}_{j,:} \y + \sum_{k=1}^{m}\Al{(k)}_{j,:}(\bias{(k)}+\lwbias{(k)}_{:,j}),
\end{equation}
where $\CROWNAu, \CROWNAl \in \R^{n_u \times n_x}$ and $\CROWNbu, \CROWNbl \in \R^{n_u}$ are defined recursively using NN weights, biases, and activations (e.g., ReLU, sigmoid, tanh) \cite{zhang2018efficient}.
\end{theorem}

% \subsection{Backward Reachability Analysis}\label{sec:backward_reachability}

% In this section, we discuss the problem of \textit{backward} reachability: given a target set $\mathcal{X}_T$, from which states will the system end up in the target set?

\subsection{Backreachable \& Backprojection Sets}

\begin{figure}[t]
\centering
\resizebox{1\columnwidth}{!}{
\begin{tikzpicture}
    % \draw [red, ultra thick, fill=red, fill opacity=0.1] (2.5,-0.5) rectangle (4,0.5) node[above, fill opacity=1] {$\mathcal{X}_{T}$};
    \draw [red, ultra thick, fill=red, fill opacity=0.1] (2.5,-0.5) rectangle (4,0.5);
    \node[red] at (3.25, 0.75) {$\mathcal{X}_{T}$};

    % Overapproximate backreachable set
    \draw [fill=white]  plot[smooth, tension=.7] coordinates {(-3.5,3.1) (-4.1,2.4) (-4.5,1.8) (-4.0,0) (-4.0,-1.8) (-2.6,-2.3) (.5,-1.7) (-0.2,0.1) (-.8,1.5) (-2.4,3.3) (-3.5,3.1)};
    \node[black] at (-2.8,2.9) {$\bar{\mathcal{R}}_{-1}(\mathcal{X}_T)$};% node[above] {$\bar{\mathcal{R}}_{-1}(\mathcal{X}_T)$};
    \node[] (i) at (3.6, 3) {};
    \node[circle,fill=black,inner sep=0pt,minimum size=3pt] (g) at (-2.1,3.) {};
    \node[circle,fill=black,inner sep=0pt,minimum size=3pt] (h) at (3.65,-0.1) {};
    %\draw [->,black, ultra thick] (g.east) to (i.west) to [out=15,in=0] (h.east);
    \draw [->,black, ultra thick] (g.east) to [out=35,in=55] (h.east);
    \node[black] (RTset) at (.0,3.8) {$\color{black} \exists \mathbf{u} \in \mathcal{U}$};

    % True Backreachable Set
    \draw [fill=black!30!green!60!white]  plot[smooth, tension=.7] coordinates {(-3,2.5) (-3.5,2.2) (-4,1.8) (-3.5,0) (-3.5,-1.8) (-2.5,-1.9) (0,-1.7) (-0.3,0.1) (-1,1.5) (-2.3,2.5) (-3,2.5)}; \node[black] at (-3,1.9) {$\color{black!40!green!90!black} \mathcal{R}_{-1}(\mathcal{X}_T)$};
    \node[circle,fill=black,inner sep=0pt,minimum size=3pt] (a) at (-2.1,2.) {};
    \node[circle,fill=black,inner sep=0pt,minimum size=3pt] (b) at (3.55,0.2) {};
    \draw [->,black!40!green!60!white, ultra thick] (a.east) to [out=25,in=55] (b.north);
    \node[black] (RTset) at (1.3,2.7) {$\color{black!40!green!90!white} \exists \mathbf{u} \in \mathcal{U}$};

    % Overapproximate backprojection
    \draw [fill=red!40!blue!20!white]  plot[smooth, tension=.7] coordinates {(-2.9,1.6) (-1.5,1.3) (-0.2,-1.0) (-1.1,-1.8) (-3.2,-1.2) (-3,0.1) (-3.7,1.1) (-2.9,1.6)}; \node[black] at (-2.5,1.) {$\color{red!60!blue!90!white} \bar{\mathcal{P}}_{-1}(\mathcal{X}_T)$};
    \node[circle,fill=black,inner sep=0pt,minimum size=3pt] (c) at (-2.3,1.4) {};
    \node[circle,fill=black,inner sep=0pt,minimum size=3pt] (d) at (2.8,0.2) {};
    \draw [->,red!55!blue!35!white, ultra thick] (c.east) to [out=350,in=120] (d.north);
    \node[black] (PTbarset) at (1, 1.65) {$ \color{red!60!blue!90!} \exists \mathbf{u} \in [\pi^L(\mathbf{x}_{-1}), \pi^U(\mathbf{x}_{-1})]$};

    % true backprojection set
    \draw [fill=blue!50!white]  plot[smooth, tension=.7] coordinates {(-2.1,0.7) (-1.1,0.45) (-0.5,-0.8) (-1.4,-1.6) (-2.8,-0.8) (-2.8,0.45) (-2.1,0.7)};
    \node[black] at (-1.9,0.25) {$\color{black} \mathcal{P}_{-1}(\mathcal{X}_T)$};
    \node[circle,fill=black,inner sep=0pt,minimum size=3pt] (e) at (-1,-0.1) {};
    \node[circle,fill=black,inner sep=0pt,minimum size=3pt] (f) at (3.0,-0.2) {};
    \draw [->,blue, ultra thick] (e.east) to [out=15,in=190] (f.west);
    \node[black] (PTset) at (1.5, -0.6) {$\color{black} \mathbf{u} = \pi(\mathbf{x}_{-1})$};

    \draw [fill=blue!20!green!20!white]  plot[smooth, tension=.7] coordinates {(-1.95,0.0) (-1.4,-0.1) (-0.6,-0.8) (-1.1,-1.4) (-2.7,-0.4) (-1.95,0.0)}; \node[black] at (-1.7, -0.5) {$\color{blue!45!green!45!black} \ubar{\mathcal{P}}_{-1}(\mathcal{X}_T)$};
    \node[circle,fill=black,inner sep=0pt,minimum size=3pt] (c) at (-1.0,-1.25) {};
    \node[circle,fill=black,inner sep=0pt,minimum size=3pt] (d) at (3.8,-.4) {};
    \draw [->,blue!50!green!50, ultra thick] (c.east) to [out=360,in=240] (d.south);
    \node[black] (PTbarset) at (2.6, -1.8) {$\color{blue!45!green!45!black} \forall \mathbf{u} \in [\pi^L(\mathbf{x}_{-1}), \pi^U(\mathbf{x}_{-1})]$};

\end{tikzpicture}}
\iffalse
\caption{{\color{red}
Backreachable, backprojection, and target sets.
Given a target set, $\mathcal{X}_T$, the backreachable set $\mathcal{R}_{-1}(\mathcal{X}_T)$ contains all states for which \textit{some} control exists to move the system to $\mathcal{X}_T$ in one timestep.
BP set $\mathcal{P}_{-1}(\mathcal{X}_T)$ contains all states for which the NN controller leads the system to $\mathcal{X}_T$.
Backreachable overapproximation $\bar{\mathcal{R}}_{-1}(\mathcal{X}_T)$ contains state for which \textit{some} states for which a control input exists to move the system to $\mathcal{X}_T$. BP under-approximation $\protect\ubar{\mathcal{P}}_{-1} (\mathcal{X}_T)$ and over-approximation $\bar{\mathcal{P}}_{-1}(\mathcal{X}_T)$ contain states for which \textit{all} and \textit{some}, respectively, controls that the relaxed NN could apply lead the system to $\mathcal{X}_T$.}
}
\fi
\caption{{\color{black}
Target, backreachable, and backprojection sets/over-approximations.
Given a target set, $\mathcal{X}_T$, the backreachable set $\mathcal{R}_{-1}(\mathcal{X}_T)$ contains all states where \textit{some} control exists such that the system goes to $\mathcal{X}_T$ in one timestep.
The BP set $\mathcal{P}_{-1}(\mathcal{X}_T)$ contains all states for which the NN controller leads the system to $\mathcal{X}_T$.
Backreachable overapproximation $\bar{\mathcal{R}}_{-1}(\mathcal{X}_T)$ contains states for which \textit{some} control exists to move the system to $\mathcal{X}_T$ under the relaxed dynamics. 
BP under-approximation $\protect\ubar{\mathcal{P}}_{-1} (\mathcal{X}_T)$ and over-approximation $\bar{\mathcal{P}}_{-1}(\mathcal{X}_T)$ contain states for which \textit{all} and \textit{some}, respectively, successor states under the relaxed system (controller, dynamics, or both) lead the system to $\mathcal{X}_T$.}
}
\label{fig:backreachable}
\vspace*{-0.12in}
\end{figure}
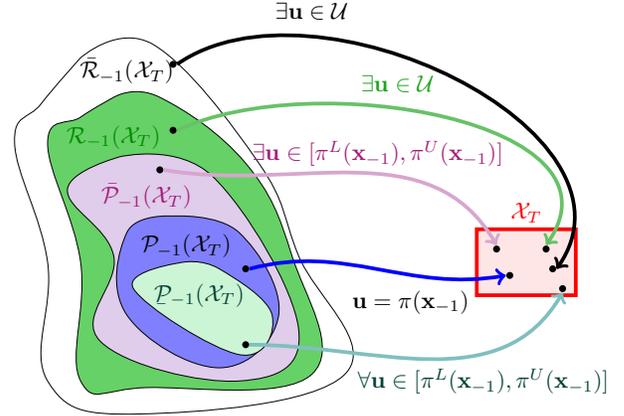

Given a target set $\mathcal{X}_T$ as shown on the right side of \cref{fig:backreachable}, each of the five sets on the left contain states that will reach $\mathcal{X}_T$ under different conditions as described below.
First, the one-step backreachable set
% \begin{align}
%     \mathcal{R}_{-1}(\mathcal{X}_T) &= \{ \mathbf{x}\ \lvert\ \exists\mathbf{u} \in \mathcal{U} \mathrm{\ s.t.\ } \label{eqn:backreachable_sets} \\ \nonumber
%      &\mathbf{A}\mathbf{x} + \mathbf{B}\mathbf{u} + \mathbf{c} \in \mathcal{X}_T \}\
% \end{align}
\begin{equation}
    \mathcal{R}_{-1}(\mathcal{X}_T) \triangleq \{ \mathbf{x}_t\ \in \mathcal{X} \lvert\ \exists\mathbf{u}_t \in \mathcal{U} \mathrm{\ s.t.\ } \label{eqn:backreachable_sets}
     p(\x_t, \mathbf{u}_t) \in \mathcal{X}_T \},
\end{equation}
contains the set of all states that transition to $\mathcal{X}_T$ in one timestep given some $\mathbf{u}_t\in \mathcal{U}$.
However, in the nonlinear case, the backreachable set cannot be computed exactly.
The dynamics $p$ are abstracted using an over-approximation and as a consequence we compute an over-approximation of the backreachable set 
\begin{align}
    \bar{\mathcal{R}}_{-1}(\mathcal{X}_T) \triangleq& \{ \mathbf{x}_t\ \in \mathcal{X} \lvert\ \exists\mathbf{u}_t \in \mathcal{U},~ \exists \x_{t+1} \bowtie \hat{p}(\x_t, \mathbf{u}_t) \mathrm{\ s.t.\ } \nonumber \\ &~~ \label{eqn:backreachable_sets_NL}
      \x_{t+1} \in \mathcal{X}_T \}
\end{align}
where a state $\x_{t+1}$ is a valid successor of $\x_t$ for some $\mathbf{u}_t \in \mathcal U$ under the abstract system $\hat p$ if $\x_{t+1} \bowtie \hat{p}(\x_t, \mathbf{u}_t)$.

The importance of the backreachable set is that it only depends on the control limits $\mathcal{U}$ and not the NN control policy $\pi$.
Thus, while $\mathcal{R}_{-1}(\mathcal{X}_T)$ is itself a very conservative over-approximation of the true set of states that will reach $\mathcal{X}_T$ under $\pi$, for the linear case it provides a region over which we can relax the NN with forward NN analysis tools, and in the nonlinear case, a less conservative region over which to approximate the dynamics and encode the NN.
%thereby avoiding issues with NN invertibility.

Next, we define the one-step true BP set as
\begin{equation}
    \mathcal{P}_{-1}(\mathcal{X}_T) \triangleq \{ \mathbf{x}_t\ \lvert\ f(\x_t;\pi) \in \mathcal{X}_T \}, \label{eqn:backprojection_sets}
\end{equation}
which denotes the set of all states that will reach $\mathcal{X}_T$ in one timestep given a control input from $\pi$.
As previously noted, calculating $\mathcal{P}_{-1}(\mathcal{X}_T)$ exactly is computationally intractable for both linear and nonlinear systems due to the presence of the NN control policy, which motivates the development of approximation techniques.
Thus, the final two sets shown in~\cref{fig:backreachable} are the BP set over-approximation and BP set under-approximation.

In the linear case, the NN control policy $\pi$ is abstracted and in the nonlinear case the dynamics $p$ are abstracted.
We generalize by denoting the abstract system $\hat f$ in both cases.
The BP set over-approximation can then be written
\begin{align}
    \bar{\mathcal{P}}_{-1}(\mathcal{X}_T) &\triangleq \{ \mathbf{x}_t\ \lvert\ \exists \x_{t+1} \bowtie \hat{f}(\x_t;\pi) \mathrm{\ s.t.\ } \label{eqn:backprojection_set_over_gen} \\ \nonumber
     & \x_{t+1} \in \mathcal{X}_T \},\
\end{align}
and the BP set under-approximation can then be written
\begin{align}
    \ubar{\mathcal{P}}_{-1}(\mathcal{X}_T) &\triangleq \{ \mathbf{x}_t\ \lvert\ \forall \x_{t+1} \bowtie \hat{f}(\x_t;\pi) \mathrm{\ s.t.\ } \label{eqn:backprojection_set_under_gen} \\ \nonumber
     & \x_{t+1} \in \mathcal{X}_T \}.\
\end{align}
%Thus, the final two sets shown in~\cref{fig:backreachable} are the BP set over-approximation
In the linear case the BP set over-approximation can be written more explicitly as
\begin{align}
    \bar{\mathcal{P}}_{-1}(\mathcal{X}_T) &\triangleq \{ \mathbf{x}\ \lvert\ \exists\mathbf{u} \in [\pi^L(\x), \pi^U(\x)] \mathrm{\ s.t.\ } \label{eqn:backprojection_set_over} \\ \nonumber
     &\mathbf{A}\mathbf{x} + \mathbf{B}\mathbf{u} + \mathbf{c} \in \mathcal{X}_T \},\
\end{align}
and the BP set under-approximation as
\begin{align}
    \ubar{\mathcal{P}}_{-1}(\mathcal{X}_T) &\triangleq \{ \mathbf{x}\ \lvert\ \forall\mathbf{u} \in [\pi^L(\x), \pi^U(\x)] \mathrm{\ s.t.\ } \label{eqn:backprojection_set_under} \\ \nonumber
     &\mathbf{A}\mathbf{x} + \mathbf{B}\mathbf{u} + \mathbf{c} \in \mathcal{X}_T \}.\
\end{align}
Comparison of the motivation behind over- and under-approximation strategies is given in the next section. From here we will omit the argument $\mathcal{X}_T$ from our notation, i.e., $\bar{\mathcal{P}}_t \triangleq \bar{\mathcal{P}}_t(\mathcal{X}_T)$, except where necessary for clarity.

\subsection{Backprojection Set: Over- vs. Under-Approximations}
{\color{black}
Given the need to approximate BP sets, it is natural to question whether it is necessary to \textit{over-} or \textit{under-}approximate the BP sets.
The ``for all'' condition (i.e., $\forall \mathbf{u} \in [\pi^L(\x), \pi^U(\x)]$) in \cref{eqn:backprojection_set_under} means that BP under-approximations can be used to provide a guarantee that a state will reach the target set since all states within the under-approximation reach the target set.
Thus, BP under-approximations are useful to check if a set of states is guaranteed to go to a goal set.
Alternatively, the ``exists'' condition (i.e., $\exists \mathbf{u} \in [\pi^L(\x), \pi^U(\x)]$) in \cref{eqn:backprojection_set_over} means that BP over-approximations can be used to check if it is possible for a state to reach the target set.
Thus, BP over-approximations are useful to check if a set of states is guaranteed to go to an obstacle/avoid set as we consider in this paper.
Note that in this work we use ``over-approximation'' and ``outer-bound'' interchangeably.
}

\section{Approach} \label{sec:approach}

\begin{figure*}[t]
\centering
\includegraphics[width=\textwidth]{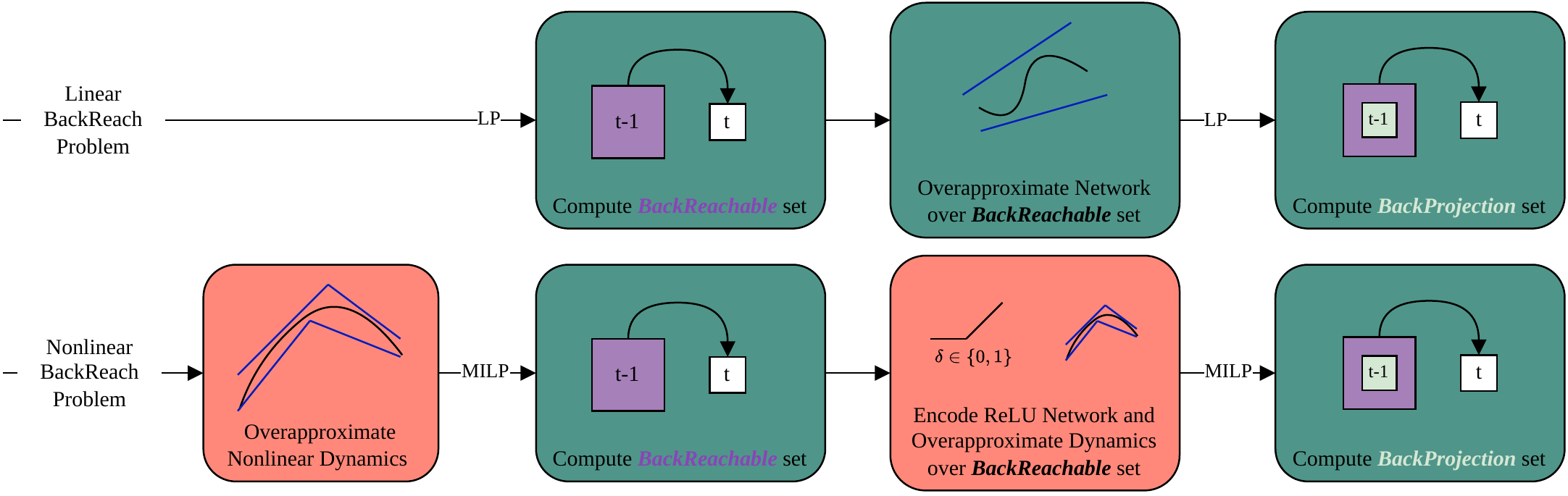}
\caption{A summary of our approach for computing backprojection sets for both linear and nonlinear dynamical systems.}
\label{fig:approach_summary}
\end{figure*}

This section outlines techniques to find multi-step BP set over-approximations $\bar{\mathcal{P}}_{-1}$ for a target set $\mathcal{X}_T$.
The key insights of computing over-approximate BP sets are captured by first discussing one-step back-projections, and then we explain the extension to multi-step BP set computation.
The general procedures for computing one-step BP sets for linear and nonlinear systems are similar and are depicted in \cref{fig:approach_summary} as well as explained in more detail as follows:
\begin{enumerate}
    \item {For smooth nonlinear dynamical systems, construct a piecewise-linear over-approximation of the dynamics, $\hat p$, over the operating domain $\mathcal{X}$}. For a linear system, $\hat p = p$.
    \item Ignoring the NN and using control limits $\mathcal{U}$ and state limits $\mathcal{X}$, find the hyper-rectangular bounds $\bar{\mathcal{R}}_{-1}$ on the true backreachable set $\mathcal{R}_{-1}$ (note that $\mathcal{P}_{-1} \subseteq \mathcal{R}_{-1} \subseteq \tilde{\mathcal{R}}_{-1} \subseteq \bar{\mathcal{R}}_{-1}$) by setting up an optimization problem with the following constraints
 
    \begin{equation}
    \mathcal{F}_{\bar{\mathcal{R}}_{-1}} \triangleq \left\{
    \mathbf{x}_t, \mathbf{u}_t\ \left| \
    \begin{aligned}
    & \x_{t+1} \bowtie \hat{p}(\x_t, \mathbf{u}_t), \\ 
    & \x_{t+1} \in \mathcal{X}_T,\\% \label{eqn:backreachable:lpmin_constr_state} \\
    & \mathbf{u}_t \in \mathcal{U}, \\ 
    & \mathbf{x}_t \in \mathcal{X}, \\ 
    \end{aligned}\right.\kern-\nulldelimiterspace
    \right\},
    \label{eqn:backreachable:constr_gen}
    \end{equation}
    and solve
    \begin{equation}
    \bar{\bar{\mathbf{x}}}_{t; \istate} = \min_{\mathbf{x}_t, \mathbf{u}_t \in \mathcal{F}_{\bar{\mathcal{R}}_{-1}} } \mathbf{e}_k^\top  \mathbf{x}_{t}, \quad \quad \ubar{\ubar{\mathbf{x}}}_{t; \istate} = \max_{\mathbf{x}_t, \mathbf{u}_t \in \mathcal{F}_{\bar{\mathcal{R}}_{-1}}} \mathbf{e}_k^\top  \mathbf{x}_{t}, \label{eqn:backreachable:obj_gen}
    \end{equation}
    for each state $\istate \in [n_x]$, where the notation $\mathbf{x}_{t;k}$ denotes the $k^\mathrm{th}$ element of $\mathbf{x}_t$ and $\mathbf{e}_k \ \in \R^{n_x}$ denotes the indicator vector, i.e., the vector with the $k^\mathrm{th}$ element equal to one and all other elements equal to zero. 
    Eq. \Cref{eqn:backreachable:obj_gen} provides a hyper-rectangular outer bound ${\bar{\mathcal{R}}_{-1}} \triangleq \{\mathbf{x}\ | \ \ubar{\ubar{\mathbf{x}}}_{t} \leq \mathbf{x} \leq \bar{\bar{\mathbf{x}}}_{t}\}$ on the backreachable set.
    
    \item Encode the neural network over the domain of the over-approximation of the backreachable set $\bar{\mathcal{R}}_{-1}$
    \label{enum:nn_relaxation} to obtain contraints corresponding to $\hat{f}(\mathbf{x}_t; \pi)$. For more detail on the differences in this step between linear and nonlinear systems, see the following two sections.
    \item Compute hyper-rectangular bounds $\bar{\mathcal{P}}_{-1}$ on the true  backprojection set $\mathcal{P}_{-1}$
    %which is the set of states that will lead to the target set when the control input is specified by the NN.
    by setting up an optimization problem with the following constraints
    \begin{equation}
    \mathcal{F}_{\bar{\mathcal{P}}_{-1}} \! \triangleq \! \left\{
    \! \mathbf{x}_{t} \left| \ 
    \begin{aligned}
    & \x_{t+1} \bowtie \hat{f}(\mathbf{x}_t; \pi) \\
    & \x_{t+1} \in \mathcal{X}_T \\ 
    & \x_t \in \bar{\mathcal{R}}_{-1} \\
    \end{aligned}\right.\kern-\nulldelimiterspace 
    \right\},
    \label{eqn:bp_gen}
    \end{equation}
    and solving the following optimization problems for each state $\istate \in [n_x]$:
    \begin{equation}
    \begin{split}
    & \bar{\bar{\mathbf{x}}}_{t; \istate} = \min_{\mathbf{x}_{t} \in \mathcal{F}_{\bar{\mathcal{P}}_{-1}}} \mathbf{e}_k^\top \mathbf{x}_{t}, \\ & \ubar{\ubar{\mathbf{x}}}_{t; \istate} = \max_{\mathbf{x}_{t} \in \mathcal{F}_{\bar{\mathcal{P}}_{-1}}} \mathbf{e}_k^\top \mathbf{x}_{t}. \label{eqn:bp_opt_gen}
    \end{split}
    \end{equation}
    where input $\mathbf{u}_t$ constraints are implicitly included by $\hat{f}(\x_t; \pi)$.
    %where $\mathbf{u}_t$ is implicitly constrained by $\hat{f}(\x_t; \pi)$.
    %The optimization problems described by \cref{eqn:bp_gen} and \cref{eqn:bp_opt_gen} are shown in \cref{fig:hybrid_constraints}. 
    %\todo{decide if we want \cref{fig:hybrid_constraints} to refer to general case or linear case.}
    
\end{enumerate}

% \begin{figure*}[t]
% % \setlength\belowcaptionskip{-0.7\baselineskip}
%     \centering
% %    \captionsetup{aboveskip=-12pt,belowskip=-12pt}
%     \includegraphics[width=1\textwidth]{figures/extension/approach.png}
%     \caption{Main idea behind our approach. \nr{Taken from a poster and added here as something we could adapt for the paper if we like}}
%     \label{fig:buggy_NN_demo}
%     \vspace*{-0.18in}
% \end{figure*}

% Detailed procedures for finding multiple-step over-approximate backprojection sets for linear dynamical systems and nonlinear dynamical systems are described in the following two sections. 
\subsection{Multi-Step BP Sets}
Now that we have explained single step BP set computation, we will explain multi-step BP set computation. 
The way multi-step BP sets are computed materially affects the speed/tractability of computation and conservativeness/tightness of the BP sets. 
Tighter sets are desirable as they more accurately reflect the system and allow proving stricter safety properties, as discussed in \cite{rober2022backward, sidrane2021overt}.
A \emph{concrete} multi-step reachability approach iteratively computes reachable sets by considering one timestep at a time:
\begin{equation}
\bar{\mathcal{P}}_t \triangleq \{\x_t \mid 
\x_{t+1} \bowtie  \hat{f}(\x_t;\pi),
\x_{t+1} \in \bar{\mathcal{P}}_{t+1} \}.
\end{equation}
If we could use the true BP sets $\mathcal P$, this would not cause any problems, but given that we are using an over-approximation of the BP set $\bar{\mathcal{P}}_{t+1}$, this  adds over-approximation error. 
In contrast, a \emph{symbolic} approach computes reachable sets by considering the constraints of several timesteps at a time:
% \begin{equation}
%     \bar{\mathcal{P}}_t \triangleq \{\x_t \mid \x_{t+1} \in 
%     \bar{\mathcal{P}}_{t+1}, \x_{t+2} \in 
%     \bar{\mathcal{P}}_{t+2}, ..., \x_{t+n} \in 
%     \mathcal{P}_{\text{final}} \}
% \end{equation}
\begin{align}
    \bar{\mathcal{P}}_t \triangleq \{\x_t \mid~ &\x_{t+1} \bowtie  \hat{f}(\x_t;\pi), \nonumber \\ 
    &\x_{t+2} \bowtie \hat{f}(\x_{t+1};\pi), ..., \nonumber \\
    &\x_{t+n} \bowtie \hat{f}(\x_{t+n-1};\pi),\nonumber \\
    &\x_{t+n} \in \mathcal{P}_{\text{final}} \}
\end{align}
where $\mathcal{P}_{\text{final}}$ may be the BP set at a point $n$ steps into the future $\bar{\mathcal{P}}_{t+n}$ or may be the target set $\mathcal{X}_T$. 
The addition of constraints results in a tighter and less conservative estimate of $\bar{\mathcal{P}}_t$ but adds to the computational complexity of the optimization problems to be solved.
The sections that follow describe in more detail both how one-step BP sets are calculated and how multi-step BP set computation is handled for both linear and nonlinear dynamical systems. 

\subsection{Linear Dynamics: Over-Approximation of BP Sets} \label{sec:approach:linear}

% We point out that while \cite{rober2022backward} iteratively calculates one-step BP over-approximations using \basic{} then refines them using \nstep{}, which implements more sophisticated LP constraints, we instead combine the two algorithms into \hybrid{} which directly uses the constraints from \nstep{} in the first calculation, thus eliminating the need for an additional refinement step.

% The general procedure described prior to this section provides the key insights to calculate a one-step BP set over-approximation.
% In this section, we outline the procedure more explicitly for linear systems, and additionally consider the calculation of a BP over-approximation at time step $t\in \mathcal{T} \triangleq \{-\tau, -\tau+1, \ldots, -1\}$ to reduce the conservativeness associated with iteratively calculating one-step BP over-approximations discussed in \cite{rober2022backward}.
In this section, we outline the procedure for calculating BP sets for linear systems.
First, to find the backreachable set for a single time step $t \in \mathcal{T} \triangleq \{ -\tau, -\tau+1, \ldots, -1 \}$, we configure an LP by writing the linear equivalent of \cref{eqn:backreachable:constr_gen} as 
\begin{equation}
\mathcal{F}_{\bar{\mathcal{R}}_t} \triangleq \left\{
\mathbf{x}_t, \mathbf{u}_t\ \left| \
\begin{aligned}
& \mathbf{A} \mathbf{x}_t + \mathbf{B} \mathbf{u}_t + \mathbf{c} \in \bar{\mathcal{P}}_{t+1},\\% \label{eqn:backreachable:lpmin_constr_state} \\
& \mathbf{u}_t \in \mathcal{U}, \\ 
& \mathbf{x}_t \in \mathcal{X}, \\ 
\end{aligned}\right.\kern-\nulldelimiterspace
\right\},
\label{eqn:backreachable:lp_constr}
\end{equation}
where $\bar{\mathcal{P}}_{0} \triangleq \mathcal{X}_{T}$, and solve
\begin{equation}
\bar{\bar{\mathbf{x}}}_{t; \istate} = \min_{\mathbf{x}_t, \mathbf{u}_t \in \mathcal{F}_{\bar{\mathcal{R}}_t}} \mathbf{e}_k^\top  \mathbf{x}_{t}, \quad \quad \ubar{\ubar{\mathbf{x}}}_{t; \istate} = \max_{\mathbf{x}_t, \mathbf{u}_t \in \mathcal{F}_{\bar{\mathcal{R}}_t}} \mathbf{e}_k^\top  \mathbf{x}_{t}, \label{eqn:backreachable:lp_obj}
\end{equation}
for each state $\istate \in [n_x]$. %, where the notation $\mathbf{x}_{t;k}$ denotes the $k^\mathrm{th}$ element of $\mathbf{x}_t$ and $\mathbf{e}_k \ \in \R^{n_x}$ denotes the indicator vector, i.e., the vector with the $k^\mathrm{th}$ element equal to one and all other elements equal to zero. 
Eq. \cref{eqn:backreachable:lp_obj} provides a hyper-rectangular outer bound $\bar{\mathcal{R}}_{t} \triangleq \{\mathbf{x}\ | \ \ubar{\ubar{\mathbf{x}}}_{t} \leq \mathbf{x} \leq \bar{\bar{\mathbf{x}}}_{t}\}$ on the backreachable set.
Note that this is a LP for the convex $\mathcal{X},\ \mathcal{U}$ used here.

Next, to compute the BP set, we need to encode the NN.
To encode the NN into an LP framework, the bounds $\ubar{\ubar{\mathbf{x}}}_{t} \leq \mathbf{x}_{t} \leq \bar{\bar{\mathbf{x}}}_{t}$ are used with \cref{thm:crown_particular_x} to construct the linear relaxation $\pi^U(\mathbf{x}_t) = \CROWNAu\mathbf{x}_t + \CROWNbu$ and $\pi^L(\mathbf{x}_t) = \CROWNAl\mathbf{x}_t + \CROWNbl$. 
Because LPs are fast to solve, a purely symbolic approach is used for multi-step BP set computation.
To compute backwards reachable sets from e.g., $t=0$ to $t=-\tau$, a set of LPs is solved for each timestep $t$ where constraints are encoded for all steps from the current timestep $t$ up to $t=0$ into each LP.
The linear, symbolic multi-step equivalent of \cref{eqn:bp_gen} is written by defining the set of state and control constraints $\mathcal{F}_{\bar{\mathcal{P}}_t}$ as
% \begin{equation}
% \mathcal{F}_{\bar{\mathcal{P}}_n} \! \triangleq \! \left\{
% \! \mathbf{x}_{\mathfrak{t}}, \mathbf{u}_{\mathfrak{t}} \left| \ 
% \begin{aligned}
% & \x_\mathfrak{t} \in \bar{\mathcal{R}}_\mathfrak{t} \\
% & \pi^L_\mathfrak{t}(\x_{\mathfrak{t}}) \leq \mathbf{u}_{\mathfrak{t}} \leq \pi^U_\mathfrak{t}(\x_\mathfrak{t}) \\
% & \x_{\mathfrak{t}+1} = \mathbf{A} \mathbf{x}_\mathfrak{t} + \mathbf{B} \mathbf{u}_\mathfrak{t} + \mathbf{c} \\
% & \x_{\mathfrak{t}+1} \in \bar{\mathcal{P}}_{\mathfrak{t}+1} \\ 
% & \mathfrak{t}\in \{-n, \ldots, -1\}
% \end{aligned} \; \;  \right.\kern-\nulldelimiterspace 
% \right\},
% \label{eqn:hybrid:lp_constr}
% \end{equation}
% and solving the following optimization problems for each state $\istate \in [n_x]$:
% \begin{equation}
% \begin{split}
% & \bar{\mathbf{x}}_{t; \istate} = \min_{\mathbf{x}_{-n:-1}, \mathbf{u}_{-n:-1} \in \mathcal{F}_{\bar{\mathcal{P}}_n}} \mathbf{e}_k^\top \mathbf{x}_{t}, \\ & \ubar{\mathbf{x}}_{t; \istate} = \max_{\mathbf{x}_{-n:-1}, \mathbf{u}_{-n:-1} \in \mathcal{F}_{\bar{\mathcal{P}}_n}} \mathbf{e}_k^\top \mathbf{x}_{t} \label{eqn:hybrid:lp_obj}
% \end{split}
% \end{equation}
\begin{equation}
\mathcal{F}_{\bar{\mathcal{P}}_t} \! \triangleq \! \left\{
\! \mathbf{x}_{i}, \mathbf{u}_{i} \left| \ 
\begin{aligned}
& \x_i \in \bar{\mathcal{R}}_i \\
& \pi^L_i(\x_{i}) \leq \mathbf{u}_{i} \leq \pi^U_i(\x_i) \\
& \x_{i+1} = \mathbf{A} \mathbf{x}_i + \mathbf{B} \mathbf{u}_i + \mathbf{c} \\
& \x_{i+1} \in \bar{\mathcal{P}}_{i+1} \\ 
& i\in \{t, \ldots, -1\}
\end{aligned} \; \;  \right.\kern-\nulldelimiterspace 
\right\},
\label{eqn:hybrid:lp_constr}
\end{equation}
and solving the following optimization problems for each state $\istate \in [n_x]$:
\begin{equation}
\begin{split}
& \bar{\mathbf{x}}_{t; \istate} = \min_{\mathbf{x}_{t:-1}, \mathbf{u}_{t:-1} \in \mathcal{F}_{\bar{\mathcal{P}}_t}} \mathbf{e}_k^\top \mathbf{x}_{t}, \\ & \ubar{\mathbf{x}}_{t; \istate} = \max_{\mathbf{x}_{t:-1}, \mathbf{u}_{t:-1} \in \mathcal{F}_{\bar{\mathcal{P}}_t}} \mathbf{e}_k^\top \mathbf{x}_{t} \label{eqn:hybrid:lp_obj}
\end{split}
\end{equation}
for each step $t \in \mathcal{T}$.
The LPs described by \cref{eqn:hybrid:lp_constr,eqn:hybrid:lp_obj} are visualized in \cref{fig:hybrid_constraints}.
Notice that for all $i > t$, the constraint $\x_i \in \bar{\mathcal{R}}_i$ is redundant because $\x_i$ is already constrained to $\bar{\mathcal{P}}_i$ by the previous time step; however, it is included for simplicity in expressing \cref{eqn:hybrid:lp_constr}.
From \cref{eqn:hybrid:lp_obj} we obtain $\bar{\mathcal{P}}_{t} \triangleq \{\mathbf{x}\ | \ \ubar{\mathbf{x}}_{t} \leq \mathbf{x} \leq \bar{\mathbf{x}}_{t}\}$ and provide the guarantee that $\mathcal{P}_{t}(\mathcal{X}_T) \subseteq \bar{\mathcal{P}}_{t}$ via the following lemma:

\begin{restatable}{lemmma}{thmhybridbackprojection}\label[lemma]{thm:hybridbackprojection}
Given an $m$-layer NN control policy $\pi:\R^{n_x}\to\R^{n_u}$, closed-loop dynamics $f: \R^{n_x} \times \Pi \to \R^{n_x}$ as in~\cref{eqn:lti_dynamics,eqn:closed_loop_dynamics}, and a convex target set $\mathcal{X}_T$, the relation
\begin{equation}
    \mathcal{P}_{t}(\mathcal{X}_T) \subseteq \bar{\mathcal{P}}_{t}(\mathcal{X}_T) \triangleq \{\x\ |\ \ubar{\x}_t \leq \x \leq \bar{\x}_t \},
\end{equation}
where $\ubar{\x}_t$ and $\bar{\x}_t$ are calculated using the LPs in \cref{eqn:hybrid:lp_obj}, holds for all $t \in \mathcal{T}$.
See~\cref{sec:proof_hybridbackprojection} for a proof.
\end{restatable}

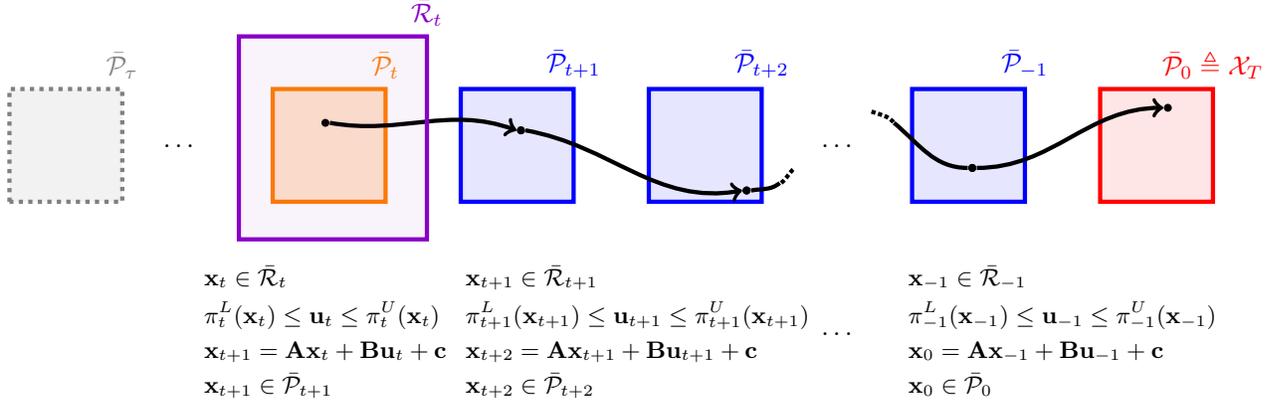
\begin{figure*}[t]
\centering
\begin{tikzpicture}

    \draw [red, ultra thick, fill=red, fill opacity=0.1] (11.5,-0.75) rectangle (13,0.75) node[above, fill opacity=1] {$\bar{\mathcal{P}}_{0} \triangleq \mathcal{X}_{T}$};

    \draw [blue, ultra thick, fill=blue, fill opacity=0.1] (9.0,-0.75) rectangle (10.5,0.75) node[above, fill opacity=1] {$\bar{\mathcal{P}}_{-1}$};
    \node[circle,fill=black,inner sep=0pt,minimum size=3pt] (c) at (9.8,-0.3) {};
    \node[circle,fill=black,inner sep=0pt,minimum size=3pt] (d) at (12.4,0.5) {};
    \draw [->,black, ultra thick] (c.east) to [out=360,in=180] (d.west);
    
    \node[black] (elipsis) at (8.0, -2.5) {$\ldots$};
    \node[black] (t0_constraints) at (11, -2.5) {\footnotesize $ \begin{aligned} & \mathbf{x}_{-1} \in \bar{\mathcal{R}}_{-1} \\ & \pi^L_{-1}(\mathbf{x}_{-1}) \leq \mathbf{u}_{-1} \leq \pi^U_{-1}(\mathbf{x}_{-1}) \\ & \mathbf{x}_{0} = \mathbf{A} \mathbf{x}_{-1} + \mathbf{B} \mathbf{u}_{-1} + \mathbf{c} \\ & \mathbf{x}_{0} \in \bar{\mathcal{P}}_{0} \end{aligned} $};

    \node[black] (elipsis) at (8.0, 0) {$\ldots$};
    \node[circle,fill=none,inner sep=0pt,minimum size=3pt] (c) at (8.4,0.45) {};
    \node[circle,fill=none,inner sep=0pt,minimum size=3pt] (d) at (8.82,0.3) {};
    \draw [densely dotted,black, ultra thick] (c.east) to [out=340,in=130] (d.west);
    \node[circle,fill=none,inner sep=0pt,minimum size=3pt] (c) at (8.7,0.3) {};
    \node[circle,fill=black,inner sep=0pt,minimum size=3pt] (d) at (9.8,-0.3) {};
    \draw [-,black, ultra thick] (c.east) to [out=320,in=180] (d.west);
    
    \draw [blue, ultra thick, fill=blue, fill opacity=0.1] (5.5,-0.75) rectangle (7,0.75) node[above, fill opacity=1] {$\bar{\mathcal{P}}_{t+2}$};
    \node[circle,fill=black,inner sep=0pt,minimum size=3pt] (c) at (6.8,-0.6) {};
    \node[circle,fill=none,inner sep=0pt,minimum size=3pt] (d) at (7.3,-0.5) {};
    \draw [-,black, ultra thick] (c.east) to [out=10,in=220] (d.west);
    \node[circle,fill=none,inner sep=0pt,minimum size=3pt] (c) at (7.2,-0.5) {};
    \node[circle,fill=none,inner sep=0pt,minimum size=3pt] (d) at (7.5,-0.3) {};
    \draw [densely dotted,black, ultra thick] (c.east) to [out=40,in=220] (d.west);
    
    \draw [blue, ultra thick, fill=blue, fill opacity=0.1] (3.,-0.75) rectangle (4.5,0.75) node[above, fill opacity=1] {$\bar{\mathcal{P}}_{t+1}$};
    \node[circle,fill=black,inner sep=0pt,minimum size=3pt] (c) at (3.8,0.2) {};
    \node[circle,fill=black,inner sep=0pt,minimum size=3pt] (d) at (6.8,-0.6) {};
    \draw [->,black, ultra thick] (c.east) to [out=350,in=190] (d.west);
    
    \draw [orange, ultra thick, fill=orange, fill opacity=0.2] (0.5,-0.75) rectangle (2.,0.75) node[above, fill opacity=1] {$\bar{\mathcal{P}}_{t}$};
    \node[circle,fill=black,inner sep=0pt,minimum size=3pt] (c) at (1.2,0.3) {};
    \node[circle,fill=black,inner sep=0pt,minimum size=3pt] (d) at (3.8,0.2) {};
    \draw [->,black, ultra thick] (c.east) to [out=350,in=160] (d.west);
    
    \node[black] (t_constraints) at (1.2, -2.5) {\footnotesize $ \begin{aligned} & \mathbf{x}_{t} \in \bar{\mathcal{R}}_{t} \\ & \pi^L_t(\mathbf{x}_t) \leq \mathbf{u}_{t} \leq \pi^U_t(\mathbf{x}_t) \\ & \mathbf{x}_{t+1} = \mathbf{A} \mathbf{x}_t + \mathbf{B} \mathbf{u}_t + \mathbf{c} \\ & \mathbf{x}_{t+1} \in \bar{\mathcal{P}}_{t+1}  \end{aligned} $};
    
    \node[black] (t_constraints) at (5.35, -2.5) {\footnotesize $ \begin{aligned} & \mathbf{x}_{t+1} \in \bar{\mathcal{R}}_{t+1} \\ & \pi^L_{t+1}(\mathbf{x}_{t+1}) \leq \mathbf{u}_{t+1} \leq \pi^U_{t+1}(\mathbf{x}_{t+1}) \\ & \mathbf{x}_{t+2} = \mathbf{A} \mathbf{x}_{t+1} + \mathbf{B} \mathbf{u}_{t+1} + \mathbf{c} \\ & \mathbf{x}_{t+2} \in \bar{\mathcal{P}}_{t+2}  \end{aligned} $};

    \draw [purple, ultra thick, fill=purple, fill opacity=0.05] (0.05,-1.25) rectangle (2.55,1.45) node[above, fill opacity=1] {$\bar{\mathcal{R}}_{t}$};
    
    \node[black] (elipsis2) at (-0.75, 0) {$\ldots$};
    
    \draw [gray, dotted, ultra thick, fill=gray, fill opacity=0.1] (-3,-0.75) rectangle (-1.5,0.75) node[above, fill opacity=1] {$\bar{\mathcal{P}}_{\tau}$};

\end{tikzpicture}
\caption{
Constraints implemented at time $t \in \mathcal{T}$ to find $\bar{\mathcal{P}}_t$ within $\bar{\mathcal{R}}_t$ by ensuring $\mathbf{x}_t$ leads to a sequence of states that can reach $\mathcal{X}_T$ with the application of control values that satisfy the relaxed neural network at each time step.
}
\label{fig:hybrid_constraints}
\vspace*{-0.12in}
\end{figure*}

\begin{algorithm}[t]
 \caption{\hybrid{}}
 \begin{algorithmic}[1]
 \setcounter{ALC@unique}{0}
 \renewcommand{\algorithmicrequire}{\textbf{Input:}}
 \renewcommand{\algorithmicensure}{\textbf{Output:}}
 \REQUIRE target state set $\mathcal{X}_T$, trained NN control policy $\pi$, time horizon $\tau$, partition parameter $r$, minimum partition element volume $v_m$
 \ENSURE BP set approximations $\bar{\mathcal{P}}_{-\tau:0}$
    \STATE $\bar{\mathcal{P}}_{-\tau:0} \leftarrow \emptyset$
    \STATE $\mathbf{\Omega}_{-\tau:-1} \leftarrow \emptyset$
    \STATE $\bar{\mathcal{P}}_{0} \leftarrow \mathcal{X}_T$
    \FOR{$t$ in $\{-1, \ldots, -\tau\}$} \label{alg:hybridBackprojection:timestep_for_loop}
        \STATE $\bar{\mathcal{R}}_{t}\! \leftarrow\! \mathrm{backreach}(\bar{\mathcal{P}}_{t\text{+}1}, \mathcal{U}, \mathcal{X})$
        \STATE $\mathcal{S} \leftarrow \mathrm{partition}(\bar{\mathcal{R}}_{t}, \mathcal{X}_T, \pi, t, r, v_m)$
        \FOR{$\mathfrak{s}$ in $\mathcal{S}$}
            \STATE $[\ubar{\ubar{\mathbf{x}}}_{t}, \bar{\bar{\mathbf{x}}}_{t}] \leftarrow \mathfrak{s}$
            \STATE $\CROWNAu, \CROWNAl, \CROWNbu, \CROWNbl \leftarrow \mathrm{CROWN}(\pi, [\ubar{\ubar{\mathbf{x}}}_{t}, \bar{\bar{\mathbf{x}}}_{t}])$
            \FOR{$\istate \in n_x$}
                \STATE $\bar{\mathbf{x}}_{t;k} \leftarrow \! \mathrm{LpMax}(\bar{\mathcal{P}}_{t+1:0}, \mathbf{\Omega}_{t:\text{-}1},\CROWNAu, \CROWNAl, \CROWNbu, \CROWNbl)$
                \STATE $\ubar{\mathbf{x}}_{t;k} \leftarrow \! \mathrm{LpMin}(\bar{\mathcal{P}}_{t+1:0},\ \mathbf{\Omega}_{t:\text{-}1},\CROWNAu, \CROWNAl, \CROWNbu, \CROWNbl)$
            \ENDFOR
            \STATE $\mathcal{A} \leftarrow \{\mathbf{x}\ \lvert\ \forall \istate \in n_x,\ \ubar{\mathbf{x}}_{t;k} \leq \mathbf{x} \leq \bar{\mathbf{x}}_{t;k}\}$
            \STATE $\bar{\mathcal{P}}_{t} \leftarrow \bar{\mathcal{P}}_{t} \cup \mathcal{A}$ \label{alg:hybridBackprojection:union}
            \STATE $\bar{\mathcal{P}}_{t} = [\ubar{\mathbf{x}}'_{t}, \bar{\mathbf{x}}'_{t}] \leftarrow \text{boundRectangle}(\bar{\mathcal{P}}_{t})$
        \ENDFOR
        \STATE $\mathbf{\Omega}_{t} = [\bar{\CROWNAu}, \bar{\CROWNAl}, \bar{\CROWNbu}, \bar{\CROWNbl}] \leftarrow \mathrm{CROWN}(\pi, [\ubar{\mathbf{x}}'_{t}, \bar{\mathbf{x}}'_{t}])$
    \ENDFOR
 \RETURN $\bar{\mathcal{P}}_{-\tau:0}$ \label{alg:hybridBackprojection:return}
 \end{algorithmic}\label{alg:hybridBackprojection}
\end{algorithm}

Due to the use of linear relaxation of the NN and axis-aligned hyper-rectangular over-approximations, the sets produced using this approach may be conservative. 
In response to this, \cite{rober2022backward} introduced the idea of partitioning the backwards reachable sets. 
Here, we combine \basic{} and \nstep{} from \cite{rober2022backward} into an algorithm that we call \hybrid{} which is shown in \cref{alg:hybridBackprojection}.
\hybrid{} implements the constraints used by \nstep{} in the same framework as \basic{}, i.e., whereas \nstep{} first calculates one-step BP over-approximations using \basic{} and adds an additional refinement step, \hybrid{} generates refined BPs during the initial calculation.
% \Cref{alg:hybridBackprojection} implements the procedure described above. 
% The algorithm is a combination of \basic{} and \nstep{} \cite{rober2022backward} that implements the constraints used by \nstep{} in the same framework as \basic{}, i.e., whereas \nstep{} first calculates one-step BP over-approximations using \basic{} and adds an additional refinement step, \hybrid{} generates refined BPs during the initial calculation.

\subsection{Efficient Partitioning for Linear Dynamics}
\label{sec:approach:partition}
% In the preceding sections we have generally treated the backreachable set as a single unit at each timestep, relaxing the NN over $\bar{\mathcal{R}}_{t}$ in its entirety.
% However, in practice, this can yield overly conservative results when taken over a large region of the state space because the policy, which can be highly nonlinear, cannot be tightly captured by affine bounds from \cref{thm:crown_particular_x}.
If the backreachable set is treated as a single unit at each timestep, relaxing the NN over $\bar{\mathcal{R}}_{t}$ in its entirety can yield overly conservative results when taken over a large region of the state space because the policy, which can be highly nonlinear, cannot be tightly captured by affine bounds from \cref{thm:crown_particular_x}.
\cref{alg:guided_partition} specifies our proposed approach for partitioning $\bar{\mathcal{R}}$ into a set of $r$ finer elements $\mathcal{S}$.
\begin{algorithm}[t]
 \caption{Guided Partitioning}
 \begin{algorithmic}[1]
 \setcounter{ALC@unique}{0}
 \renewcommand{\algorithmicrequire}{\textbf{Input:}}
 \renewcommand{\algorithmicensure}{\textbf{Output:}}
 \REQUIRE backreachable set $\bar{\mathcal{R}}$, target set $\mathcal{X}_T$, trained NN control policy $\pi$, time to target $t$, number partitions $r$, minimum element volume $v_m$
 \ENSURE list of backreachable set elements $\mathcal{S}$
 
 \STATE $\mathcal{S} \leftarrow \mathcal{R}$
 \STATE $\mathcal{Q} \leftarrow \mathrm{estimateBpSet}(\mathcal{X}_T, t)$  \label{alg:guided_partition:estimate}
 
 \WHILE{$\mathcal{S}[\mathrm{end}]$.shouldBeSplit and length($\mathcal{S}) < r$}
    \STATE $\mathfrak{s} \leftarrow \mathcal{S}.\mathrm{pop()}$
    \STATE $\mathfrak{s}_1, \mathfrak{s}_2 \leftarrow \mathrm{split}(\mathfrak{s}, \mathcal{Q})$
    \FOR{$\mathfrak{s}_i$ in $[\mathfrak{s}_1, \mathfrak{s}_2]$}
        \STATE $\mathfrak{s}_i.\mathrm{value} = \mathrm{L1norm}(\mathfrak{s}_i, \mathcal{Q})$
        \IF{satisfiable($\mathfrak{s}_i, \mathcal{F}_{\bar{\mathcal{P}}_t}$) and $\mathfrak{s}_i$.volume $> v_{m}$}
            \STATE $\mathfrak{s}_i.\mathrm{shouldBeSplit} \leftarrow$ True
        \ELSE
            \STATE $\mathfrak{s}_i.\mathrm{shouldBeSplit} \leftarrow$ False
            \STATE $\mathfrak{s}_i.\mathrm{value} = 0$
        \ENDIF
        \STATE $\mathcal{S}.\mathrm{insert}(\mathfrak{s}_i)$
    \ENDFOR

 \ENDWHILE
 
 \RETURN $\mathcal{S}$ \label{alg:guided_partition:return}
 \end{algorithmic}\label{alg:guided_partition}
\end{algorithm}

We first initiate the list $\mathcal{S}$ with the full backreachable set $\bar{\mathcal{R}}$. \cref{alg:guided_partition:estimate} then generates an estimate $\mathcal{Q}$ of the true BP set by randomly sampling points over $\bar{\mathcal{R}}$ and providing a rectangular bound of the points that reach the target set. 
%Note that while the estimation step is not required, it is inexpensive and helps the partitioner determine which areas need to be partitioned more finely. 
Note that because $\mathcal{Q}$ is an under-approximation of $\mathcal{P}$, there is no reason to partition within $\mathcal{Q}$ because $\mathcal{Q} \subset \bar{\mathcal{P}}$ regardless of how tight our NN relaxation is.
Next, the last element (at first, the only element is $\bar{\mathcal{R}}$) is popped from the end of $\mathcal{S}$ and split such that: if it contains samples from \cref{alg:guided_partition:estimate}, all the samples are segregated into only one of the resulting elements.
If no samples are in the element, it is simply split in half.
The two new elements are then added back to the list and ordered by their distance from $\mathcal{Q}$.
This process is repeated until the partitioning budget is met, or if the last element in $\mathcal{S}$ should not be split (either because it is below the minimum volume $v_m$ or does not contain any $\mathbf{x}$ that satisfies \cref{eqn:hybrid:lp_constr}.
\cref{fig:partitioning_step} shows the beginning of an execution of \cref{alg:guided_partition} with a visual illustration of how the {\tt split} and {\tt L1norm} functions work.
\begin{figure}[t]
    \centering
%    \captionsetup{aboveskip=-12pt,belowskip=-12pt}
    \includegraphics[width=1\columnwidth]{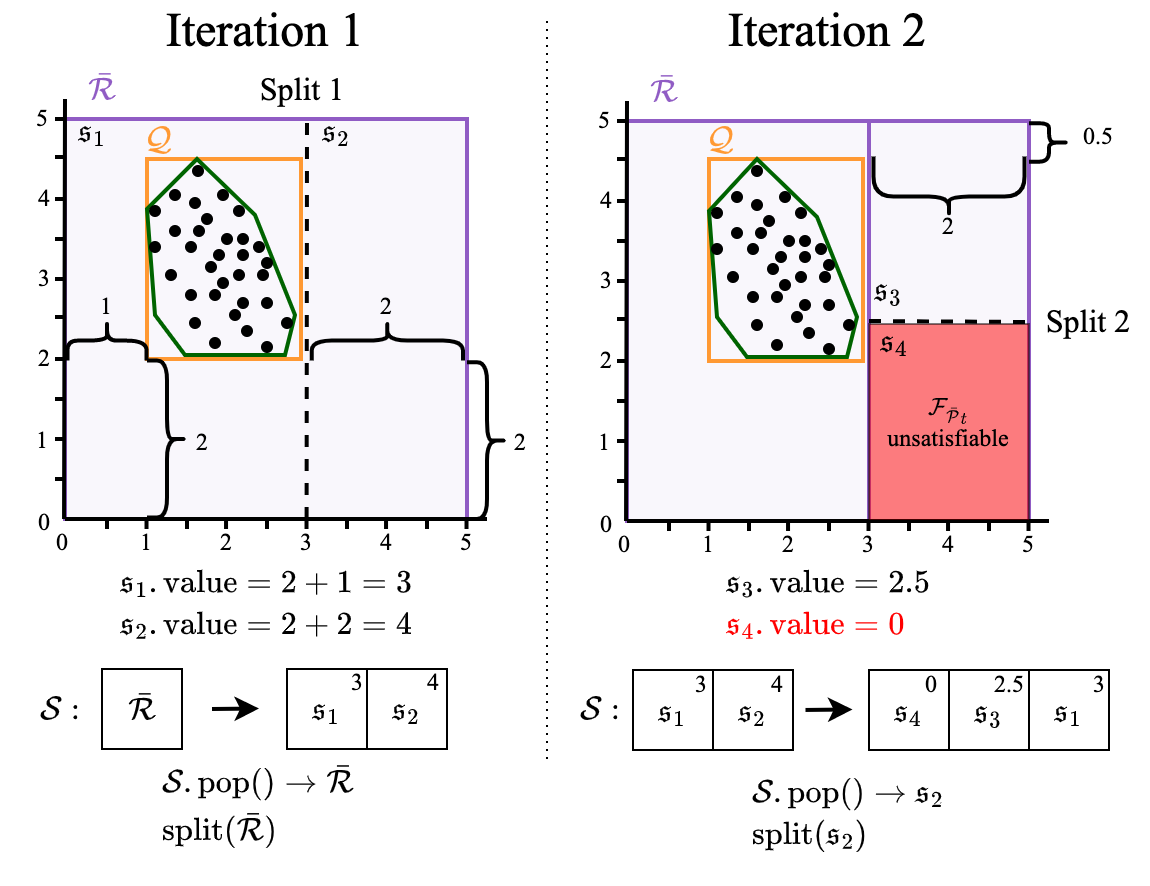}
    \caption{Two iterations of guided partitioning. In the first iteration, a split is made such that all the Monte-Carlo samples (black) that under-approximate $\mathcal{P}$ are contained in $\mathfrak{s}_1$. In the second iteration, $\mathfrak{s}_2$ is split because it is furthest from the MC sample bound $\mathcal{Q}$, and one resulting element $\mathfrak{s}_4$ cannot satisfy $\mathcal{F_{\bar{P}}}_t$. The update made to the element queue $\mathcal{S}$ is displayed for each iteration.}
    \label{fig:partitioning_step}
    \vspace*{-0.18in}
\end{figure}

Finally, note that \cref{alg:hybridBackprojection} requires solving two LPs for each state.
From this, it is apparent that the number of LPs solved $N_{LP}$ can be written as $N_{LP} = 2 n_x r \tau$, i.e., two LPs per state for each partition over the time horizon $\tau$.
However, as $\bar{\mathcal{P}}_t$ is built by adding each element's contribution $\mathcal{A}$ (\cref{alg:hybridBackprojection:union}), there may be LPs that cannot possibly change $\bar{\mathcal{P}}_t$ because $\ubar{\ubar{\mathbf{x}}}_t \in \bar{\mathcal{P}}_t$ or $\bar{\bar{\mathbf{x}}}_t \in \bar{\mathcal{P}}_t$. 
In such cases, we can skip the irrelevant LP, thus reducing computational cost as will be shown in \cref{sec:results}, where we refer to this strategy as \textit{SkipLP}.

\subsection{Nonlinear Dynamics: Over-Approximation of BP Sets}
\Cref{alg:NLBackprojection} overviews the steps for computing backwards reachable sets when the problem dynamics are nonlinear. 
\BRNL{} describes a \emph{concrete} approach to multi-step BP set computation where the constraints corresponding to only 1 timestep are encoded at a time. % A \emph{symbolic} or \emph{hybrid-symbolic}~\cite{sidrane2021overt} approach is also possible and is described in \cref{sec:hybridsym}.
A \emph{symbolic} or \emph{hybrid-symbolic}~\cite{sidrane2021overt} approach is also possible and is described in the next section.
%because a fully concrete approach produced tight backwards reachable sets and hybrid-symbolic approaches were not computationally tractable.
First, backreachable sets are computed.
To make the optimization problem for computing the backreachable set tractable for nonlinear dynamics, we over-approximate the dynamics using OVERT \cite{sidrane2021overt} to obtain $\hat p$. 
With this over-approximation, the first constraint in \cref{eqn:backreachable:constr_gen} becomes a set of mixed integer constraints, and we can solve the optimization problem in \cref{eqn:backreachable:obj_gen} using a MILP solver to obtain an over-approximation of the backreachable set.

Given the backreachable set, we recompute the over-approximated dynamics model $\hat p$ over this smaller domain to minimize over-approximation error and decrease the number of total constraints in the optimization problem. By combining these constraints with a mixed integer encoding of the control network $\pi$ using the techniques of \citet{tjeng2017evaluating}, we obtain a set of mixed integer constraints corresponding to the abstract closed-loop system $\hat f$. The optimization problem in \cref{eqn:bp_opt_gen} then becomes a MILP and can be solved using a MILP solver to obtain an over-approximated BP set.

\begin{algorithm}[t]
 \caption{\BRNL}
 \begin{algorithmic}[1]
 \setcounter{ALC@unique}{0}
 \renewcommand{\algorithmicrequire}{\textbf{Input:}}
 \renewcommand{\algorithmicensure}{\textbf{Output:}}
 \REQUIRE target state set $\mathcal{X}_T$, trained NN control policy $\pi$, time horizon $\tau$, dynamics model $p$, $\mathrm{OVERT}$ tolerance $\epsilon$
 \ENSURE BP set approximations $\bar{\mathcal{P}}_{-\tau:0}$
    \STATE $\bar{\mathcal{P}}_{-\tau:0} \leftarrow \emptyset$
    \STATE $\bar{\mathcal{P}}_{0} \leftarrow \mathcal{X}_T$
    \STATE $\hat p_\mathcal{X} \leftarrow \mathrm{OVERT}(p, \mathcal X, \mathcal U, \epsilon)$
    \FOR{$t$ in $\{-1, \ldots, -\tau\}$} \label{alg:NLBackprojection:timestep_for_loop}
        %\STATE $\hat p \leftarrow \mathrm{OVERT}(p, \mathcal X, \mathcal U, \epsilon)$
        \STATE $\bar{\mathcal{R}}_{t}\! \leftarrow\! \mathrm{backreach}(\bar{\mathcal{P}}_{t\text{+}1}, \mathcal{U}, \mathcal{X}, \hat p_\mathcal{X})$
        \STATE $\hat p \leftarrow \mathrm{OVERT}(p, \bar{\mathcal{R}}_{t}, \mathcal U, \epsilon)$
        \STATE $\hat f \leftarrow \mathrm{MIP}(\hat p, \pi)$
        \FOR{$\istate \in n_x$}
            \STATE $\bar{\mathbf{x}}_{t;k} \leftarrow \! \mathrm{MILPMax}(\bar{\mathcal{P}}_{t+1}, \hat f)$
            \STATE $\ubar{\mathbf{x}}_{t;k} \leftarrow \! \mathrm{MILPMin}(\bar{\mathcal{P}}_{t+1}, \hat f)$
        \ENDFOR
        \STATE $\bar{\mathcal{P}}_{t} \leftarrow \{\mathbf{x}\ \lvert\ \forall \istate \in n_x,\ \ubar{\mathbf{x}}_{t;k} \leq \mathbf{x} \leq \bar{\mathbf{x}}_{t;k}\}$
    \ENDFOR
 \RETURN $\bar{\mathcal{P}}_{-\tau:0}$ 
 \end{algorithmic}\label{alg:NLBackprojection}
\end{algorithm}

\subsection{Symbolic and Hybrid-Symbolic Methods}\label{sec:hybridsym}
The \textit{concrete} reachability method described in \cref{alg:NLBackprojection} compounds over-approximation error over time by creating a hyperrectangular over-approximation for each time step. A \textit{symbolic} approach can alleviate this compounding. To perform symbolic reachability for systems with nonlinear dynamics, we write a symbolic equivalent to \cref{eqn:bp_gen} as follows
\begin{equation}
\mathcal{F}_{\bar{\mathcal{P}}_t} \! \triangleq \! \left\{
\! \mathbf{x}_{i} \left| \ 
\begin{aligned}
& \x_i \in \bar{\mathcal{R}}_i \\
& \x_{i+1} \bowtie \hat{f}(\mathbf{x}_i; \pi) \\
& \x_{i+1} \in \bar{\mathcal{P}}_{i+1} \\ 
& i\in \{t, \ldots, -1\}
\end{aligned} \; \;  \right.\kern-\nulldelimiterspace 
\right\}.
\label{eqn:sym:constr}
\end{equation}
Symbolic reachability provides tighter bounds on the BP set at the expense of a higher computational cost. For this reason, it is desirable to balance between tightness and computational efficiency. We therefore take a \textit{hybrid-symbolic} approach to nonlinear backwards reachability similar to the approach described in \citet{sidrane2021overt}. In this approach, we reduce compounded over-approximation error by  computing a symbolic set after computing a given number of concrete sets, and repeating this process iteratively.
% Concrete steps can accumulate overapproximation error but are more computationally efficient
% Symbolic can be done with constraints
% More computation required
% Might want balance, where only run symbolic every few steps
% Hybrid approach

\subsection{Sources of Over-Approximation Error}
To produce useful estimates for safety analysis and limit over-conservative bounds, it is important to understand and mitigate the effect of over-approximation error on the estimated BP sets. 
% Both \cref{alg:hybridBackprojection} and \cref{alg:NLBackprojection} have two sources of over-approximation error, although they differ in their first source. 
For systems with linear dynamics (\cref{alg:hybridBackprojection}), the first source of over-approximation error is the linear relaxation of the NN using \cref{thm:crown_particular_x} to be used as constraints in the LP. 
Partitioning the backreachable sets, as described in \cref{sec:approach:partition}, can reduce this error. 
For systems with nonlinear dynamics, on the other hand, since we solve a MILP rather than LP to compute BP sets (\cref{alg:NLBackprojection}), we can encode the network exactly using mixed-integer constraints \cite{tjeng2017evaluating}. 
Therefore, the network encoding does not contribute to over-approximation error in the nonlinear case. 
Instead, the first source of over-approximation error in \cref{alg:NLBackprojection} is in the overapproximation of the nonlinear dynamics using piecewise-linear functions. 
To minimize this contribution, the piecewise-linear functions are selected to produce optimally tight bounds given a user-specified error tolerance.

The second source of over-approximation error is present in both algorithms and results from \cref{eqn:bp_opt_gen}, which produces axis-aligned, hyper-rectangular BP sets. 
The use of hyper-rectanglar sets can result in large over-approximation errors if the true BP sets are not well-represented by axis-aligned hyper-rectangles. 
This type of error may become especially apparent when using a concrete approach as in \cref{alg:NLBackprojection}.
%since the hyper-rectangular BP sets are made concrete at each time step and used as the target set for the next time step. 
Concretizing the BP set at each time step results in a compounding of the hyper-rectangular over-approximation error. 
In contrast, a symbolic approach as in \cref{alg:hybridBackprojection} encodes the dynamics constraints for all time steps in the optimization and only solves for the hyper-rectangular bounds at the time step of interest. 
Only concretizing the final BP set alleviates the compounding error at the expense of a more complex optimization problem.
Partitioning the backwards reachable sets as in \hybrid{} also reduces the over-approximation error due to the use of hyper-rectangular sets.

% Two main sources of overapproximation in each case
% For the linear case, overapproximation comes from the crown bounds
% For the nonlinear case, comes from OVERT bounds (no overappox of network)
% For both cases, representing as hyperrectangles causes overapprox
    % Symbolic helps

%!TEX root=main.tex

\section{Numerical Results}
\label{sec:results}
In this section, we use numerical experiments to demonstrate \hybrid{} and \BRNL{} and compare their properties to the algorithms proposed in \cite{rober2022backward} and to the forward reachability tool proposed in \cite{everett2021reachability}.
First, we show how \hybrid{} improves upon \basic{} and \nstep{} from \cite{rober2022backward} and demonstrate how the guided partitioning strategy proposed in \cref{alg:guided_partition} can outperform a uniform partitioning strategy.
We then demonstrate how \hybrid{} can be used to verify safety in a collision avoidance scenario that causes the forward reachability tool Reach-LP \cite{everett2021reachability} to fail. Next, we use \BRNL{} to verify safety of nonlinear version of the same collision avoidance scenario.
Finally, we conduct an ablation study to show how \hybrid{} improves the work done in \cite{rober2022backward}, enabling more efficient verification of an obstacle avoidance scenario for a 6D linearized quadrotor and discuss how the proposed algorithms scale with system dimension.

All numerical results for systems with linear dynamics were collected with the LP solver {\tt cvxpy} \cite{diamond2016cvxpy} on a machine running Ubuntu 20.04 with an i7-6700K CPU and 32 GB of RAM. For the experiments with nonlinear dynamics, results were obtained using the Gurobi \cite{gurobi} MILP solver on single Intel Core i7 4.20 GHz processor with 32 GB of RAM.

\subsection{Double Integrator}

Consider the discrete-time double integrator model \cite{hu2020reach}
\begin{equation}
    \mathbf{x}_{t+1} =
    \underbrace{
    \begin{bmatrix}
    1 & 1 \\
    0 & 1
    \end{bmatrix}}_{\mathbf{A}} \mathbf{x}_t +
    \underbrace{
    \begin{bmatrix}
    0.5 \\ 1
    \end{bmatrix}}_{\mathbf{B}} \mathbf{u}_t
\end{equation}
with $\mathbf{c}=0$, and discrete sampling time $t_s=1$s. 
The NN controller (identical to the double integrator controller used in \cite{everett2021reachability}) has $[5,5]$ neurons, ReLU activations and was trained with state-action pairs generated by an MPC controller.
\subsubsection{Update vs. Previous Approach}
\cref{fig:double_integrator_n_vs_one}  compares \basic{} (orange) and \nstep{} (blue) from \cite{rober2022backward} to \hybrid{} (dashed). 
As shown, each algorithm collects some approximation error
\begin{equation}
    \mathrm{error} = \frac{A_{BPE}-A_\mathrm{true}}{A_\mathrm{true}},
    \label{eqn:BP_estimate_error}
\end{equation}
where $A_\mathrm{true}$ denotes the area of the tightest rectangular bound of the true BP set (dark green in~\cref{fig:double_integrator_n_vs_one:BP_sets}), calculated using Monte Carlo simulations, and $A_\mathrm{BPE}$ denotes the area of the BP estimate.
As shown by \cref{fig:double_integrator_n_vs_one} and \cref{tab:alg_error_comparison}, \hybrid{} falls between \basic{} and \nstep{} in terms of final step error and computation time. However, in addition to guided partitioning and the SkipLP strategy described in \cref{sec:approach:partition}, we also improved computational efficiency from \cite{rober2022backward} by implementing disciplined parametric programming (DPP) \cite{diamond2016cvxpy} which allows {\tt cvxpy} to solve large amounts of LPs more efficiently.
With DPP, the time discrepancy between \basic{} and \hybrid{} disappears as shown by the third column of  \cref{tab:alg_error_comparison}, thus leading to \hybrid{} as the go-to strategy.
\begin{figure}[t]
\setlength\belowcaptionskip{-0.1\baselineskip}
\centering
\captionsetup[subfigure]{aboveskip=-1pt,belowskip=-1pt}
    \begin{subfigure}[t]{\columnwidth}
        \begin{tikzpicture}[fill=white]
            \node[anchor=south west,inner sep=0] (image) at (0,0) {\includegraphics[width=1\columnwidth]{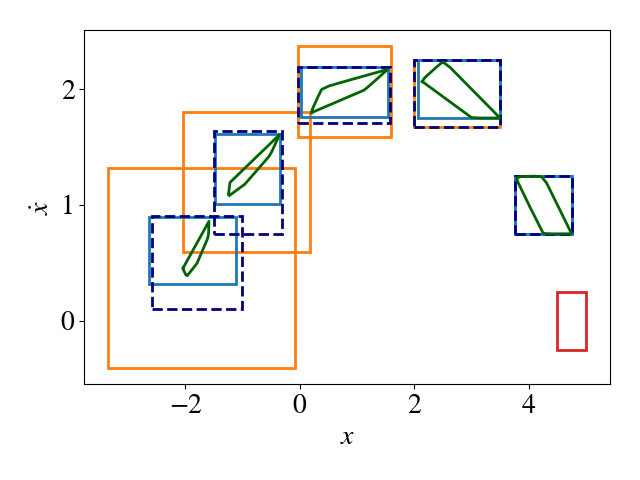}};
            \begin{scope}[x={(image.south east)},y={(image.north west)}]
                \node[] at (0.85,.23) {{\small Target Set}};
            \end{scope}
        \end{tikzpicture}
        \caption{BP set estimates and true BP convex hulls (dark green)}
        \label{fig:double_integrator_n_vs_one:BP_sets}
    \end{subfigure}
    \begin{subfigure}[t]{\columnwidth} 
        \begin{tikzpicture}[fill=white]
            \node[anchor=south west,inner sep=0] (image) at (0,0) {\includegraphics[width=1\columnwidth]{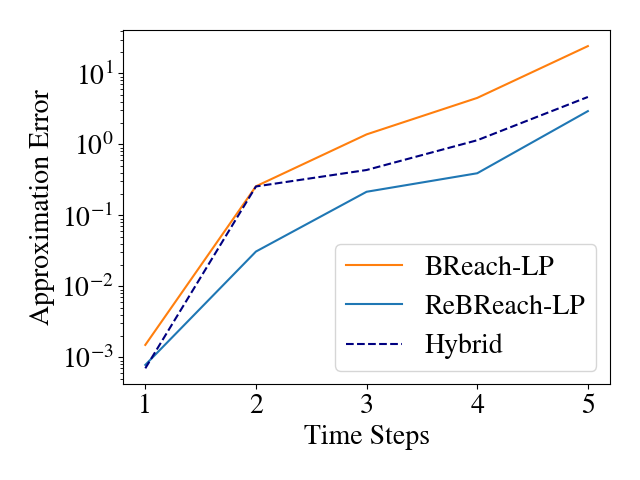}};
            % \begin{scope}[x={(image.south east)},y={(image.north west)}]
            %     \fill [white] (0.57,0.9) rectangle (0.35,0.74);
            % \end{scope}
            % \begin{scope}[x={(image.south east)},y={(image.north west)}]
            %     \node[] at (0.455,0.86) {{\tiny \basic{}}};
            %     % \draw (0.44, 0.5) \node{$t=1$}; 
            % \end{scope}
            % \begin{scope}[x={(image.south east)},y={(image.north west)}]
            %     \node[text width=2cm,align=center] at (0.455,0.78) {{\tiny \nstep{}}};
            %     % \draw (0.44, 0.5) \node{$t=1$}; 
            % \end{scope}
            
        \end{tikzpicture}
        \caption{Approximation error \cref{eqn:BP_estimate_error} calculated at each time step}
        \label{fig:double_integrator_n_vs_one:error_comparison}
    \end{subfigure}
    \caption{BP set estimates for a double integrator extending from the target set (red) calculated with \basic{} (orange) and \nstep{} (blue) and \hybrid{} (dashed).}
    \label{fig:double_integrator_n_vs_one}
\end{figure}

\begin{table}[t]
\centering
\caption{Compare error~\cref{eqn:BP_estimate_error} for BP algorithms.}

\begin{tabular}{@{}lccr@{}}
\toprule
& \multicolumn{2}{c}{Runtime [s]} & \\
\cmidrule(lr){2-3}
 Algorithm   & \cite{rober2022backward} & DPP+SkipLP &  Final Error \\
\midrule
 BReach-LP   & $1.540 \pm 0.062$ & $0.855 \pm 0.014$ &   24.5  \\
 \rowcolor{green} \hybrid{}      & $2.243 \pm 0.043$ & $0.838 \pm 0.040$ &    4.68 \\
 ReBReach-LP & $3.410 \pm 0.036$ & $1.425 \pm 0.043$ &    2.97 \\
\bottomrule
\end{tabular}
\label{tab:alg_error_comparison}
\vspace*{-0.12in}
\end{table}

% \begin{table}[t]
% \centering
% \caption{Replication of \cref{tab:alg_error_comparison:old} with LPs calculated using DPP, leading to a computational speedup.}
% \begin{tabular}{llr}
% \hline
%  Algorithm   & Runtime [s]       &   Final Step \\ & & Error \\
% \hline
%  BReach-LP   & $0.855 \pm 0.014$ &              24.5  \\
%  Hybrid      & $0.838 \pm 0.040$ &               4.68 \\
%  ReBReach-LP & $1.425 \pm 0.043$ &               2.97 \\
% \hline
% \end{tabular}
% \label{tab:alg_error_comparison:new}
% \vspace*{-0.12in}
% \end{table}

% \begin{table}[t]
% \centering
% \caption{Replication of \cref{tab:alg_error_comparison:old} with LPs calculated using DPP, leading to a computational speedup.}
% \begin{tabular}{lllr}
% \hline
%  Algorithm   & Runtime [s]       &   Final Step Error & t \\
% \hline
%  BReach-LP   & $0.855 \pm 0.014$ & 24.50 & e \\
%  Hybrid      & $0.838 \pm 0.040$ &  4.68 & s \\
%  ReBReach-LP & $1.425 \pm 0.043$ &  2.97 & t \\
% \hline
% \end{tabular}
% \label{tab:alg_error_comparison:new}
% \vspace*{-0.12in}
% \end{table}

\subsubsection{Uniform vs. Guided Partitioning}

Next, we use the double integrator model to visualize how the partitioning strategy introduced in \cref{alg:guided_partition} compares to the uniform strategy used in \cite{rober2022backward}. \cref{fig:guided_vs_uniform_sets} first shows how each partitioning strategy performs when budgeted with different numbers of partitions. Comparing \cref{fig:guided_vs_uniform_sets:3x3} and \cref{fig:guided_vs_uniform_sets:9} shows that the guided partitioning strategy is able to converge faster than uniform partitioning. This is because guided partitioning prioritizes regions that are causing the most conservativeness in the BP set estimate. Comparing \cref{fig:guided_vs_uniform_sets:16x16} and \cref{fig:guided_vs_uniform_sets:256} demonstrates how the guided partitioning strategy is able to avoid wasting time partitioning regions of the backreachable set that are far from the true BP set.

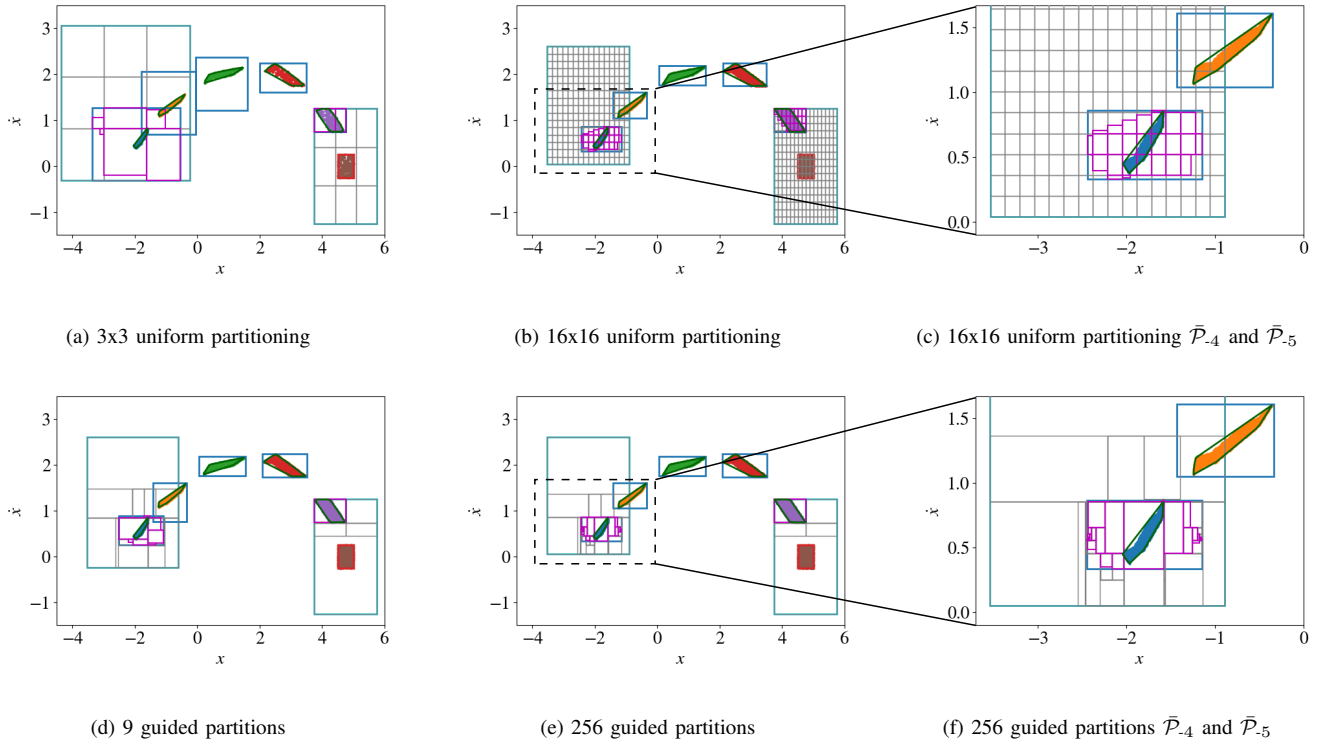
\begin{figure*}%
%% first three subfigures
\centering
\begin{subfigure}{.31\textwidth}
  \label{fig:Total_scatter}%
  \begin{tikzpicture}
      \node{\includegraphics[width=1\textwidth]{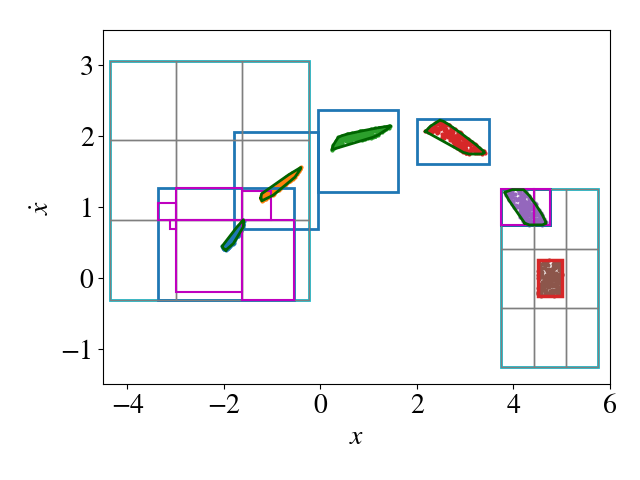}};
  \end{tikzpicture}%
  \caption{3x3 uniform partitioning}\label{fig:guided_vs_uniform_sets:3x3}
\end{subfigure}
\hspace*{\fill}
\begin{subfigure}{.31\textwidth}
  \label{fig:Num_scatter}%
  \begin{tikzpicture}[remember picture]
      \node{\includegraphics[width=1\textwidth]{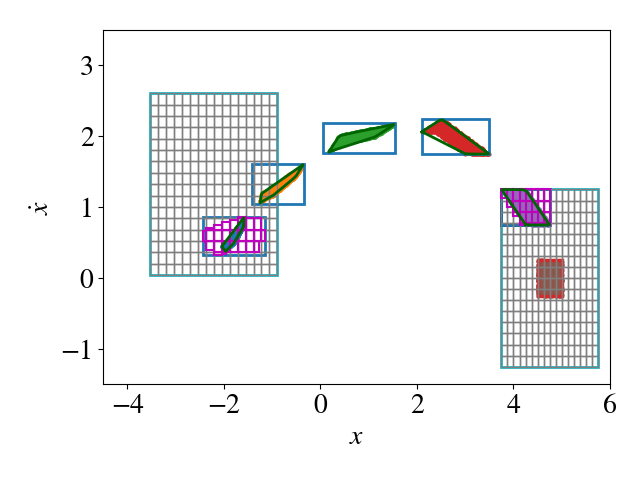}};
    %   \node[] at (0.2,-0.3) {\footnotesize{(c)}};
      \draw [draw=black,dashed,line width=0.2mm] (-1.62,-.42) rectangle ++(1.6,1.12);
    %   \draw [draw=black] (-1.86,-1.25) rectangle ++(4.37,3.05);
  \end{tikzpicture}
  \caption{16x16 uniform partitioning}\label{fig:guided_vs_uniform_sets:4x4}
\end{subfigure}
\hspace*{\fill}
\begin{subfigure}{.31\textwidth}
  \label{fig:Raw_scatter}%
  \begin{tikzpicture}[remember picture]
      \node{\includegraphics[width=1\textwidth]{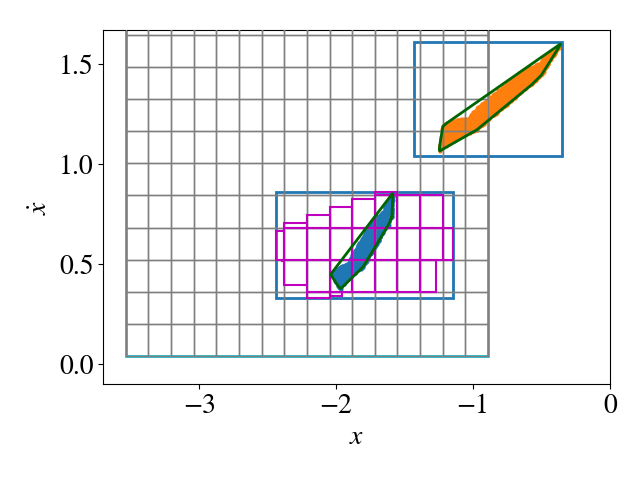}};
  \end{tikzpicture}%
  \caption{16x16 uniform partitioning $\bar{\mathcal{P}}_{\text{-}4}$ and $\bar{\mathcal{P}}_{\text{-}5}$ }\label{fig:guided_vs_uniform_sets:16x16}
\end{subfigure}
%% second group of subfigures
\begin{subfigure}{.31\textwidth}
  \label{fig:Num_Raw_scatter}%
  \begin{tikzpicture}
      \node{\includegraphics[width=1\textwidth]{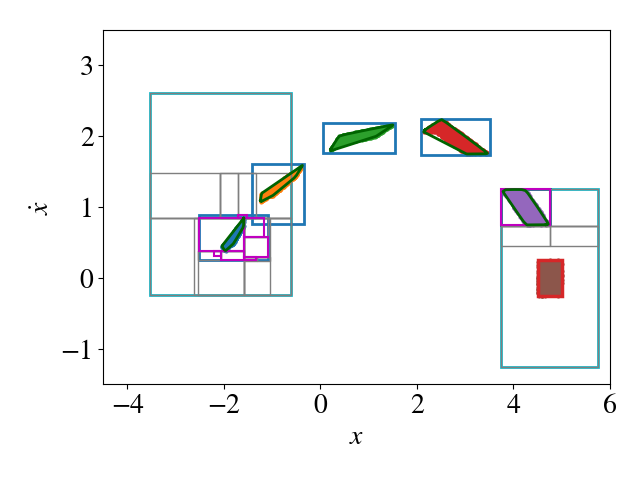}};
  \end{tikzpicture}%
  \caption{9 guided partitions}
  \label{fig:guided_vs_uniform_sets:9}
\end{subfigure}
\hspace*{\fill}
\begin{subfigure}{.31\textwidth}
  \label{fig:Total_Raw_scatter}%
  \begin{tikzpicture}
      \node{\includegraphics[width=1\textwidth]{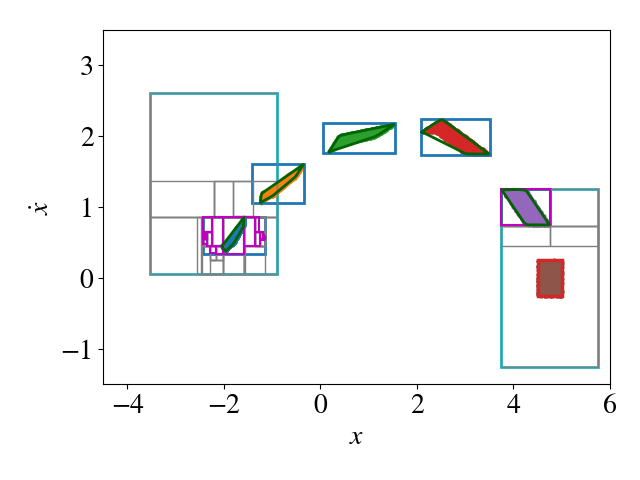}};
    %   \node[] at (0.2,-0.3) {\footnotesize{(f)}};
      \draw [draw=black,dashed,line width=0.2mm] (-1.62,-.42) rectangle ++(1.6,1.12);
  \end{tikzpicture}
  \caption{256 guided partitions}\label{fig:guided_vs_uniform_sets:16}
\end{subfigure}
\hspace*{\fill}
\begin{subfigure}{.31\textwidth}
  \label{fig:Num_Total_scatter}%
  \begin{tikzpicture}
      \node{\includegraphics[width=1\textwidth]{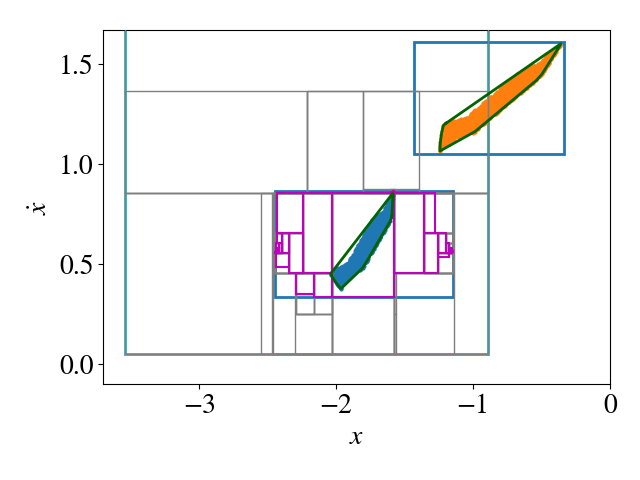}};
  \end{tikzpicture}
  \caption{256 guided partitions $\bar{\mathcal{P}}_{\text{-}4}$ and $\bar{\mathcal{P}}_{\text{-}5}$ }\label{fig:guided_vs_uniform_sets:256}
\end{subfigure}

\begin{tikzpicture}[overlay, remember picture]
    \draw[line width=0.2mm] (.1,7.8) -- (4.36,6.98);
    \draw[line width=0.2mm] (.1,8.92) -- (4.36,10.01);
    \draw[line width=0.2mm] (.1,2.6) -- (4.36,1.78);
    \draw[line width=0.2mm] (.1,3.72) -- (4.36,4.81);
\end{tikzpicture}

\caption{Partitioning of $\bar{\mathcal{R}}_t$ for first and last timesteps showing that the guided partitioning strategy disregards irrelevant regions while the uniform partitioner wastes time refining regions that cannot affect $\bar{\mathcal{P}}_t$.}\label{fig:guided_vs_uniform_sets}
\end{figure*}

Next, \cref{fig:partitioning_time_comparison} shows how the guided strategy is able to reduce approximation error more quickly than the uniform strategy. Notice that the guided strategy terminates the partitioning algorithm when it gets to the ``best possible" solution, obtained using $100 
\times 100$ uniform partitions.
\begin{figure}[t]
    \centering
%    \captionsetup{aboveskip=-12pt,belowskip=-12pt}
    \includegraphics[width=1\columnwidth]{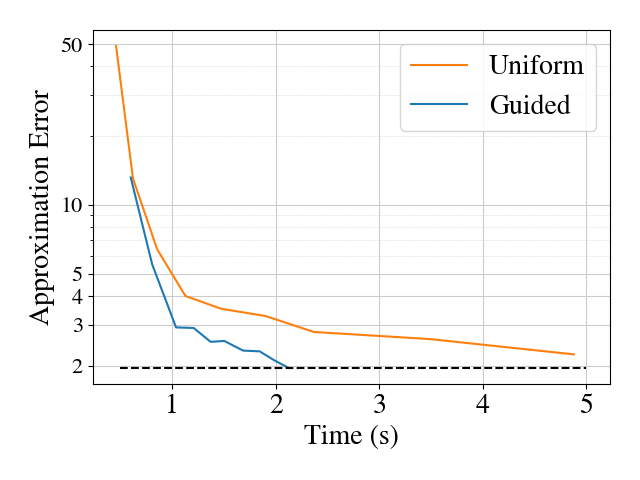}
    \caption{Guided partitioning converges to the ``best possible" approximation error in 2s, whereas uniform partitioning still has not converged using $12 
    \times 12$ partitioning with 5s of calculation time.}
    \label{fig:partitioning_time_comparison}
    \vspace*{-0.18in}
\end{figure}
To explain why the best possible error using our method is not 0, we show in \cref{fig:lower_bound_explanation} an example of a point that satisfies the conditions from \cref{eqn:hybrid:lp_constr}, but does not lie among the samples representing the true BP set. 
\cref{fig:lower_bound_explanation} shows that because the BP sets are approximated with an axis-aligned hyper-rectangle, it can be easy to step through $\bar{\mathcal{P}}$, but not through $\mathcal{P}$, leading to some conservativeness.
This issue highlights the need to adapt our methods to accommodate other set representations that can more tightly approximate true BP sets, which is an interesting and practical direction for future work.
\begin{figure}[t]
    \centering
%    \captionsetup{aboveskip=-12pt,belowskip=-12pt}
    \includegraphics[width=1\columnwidth]{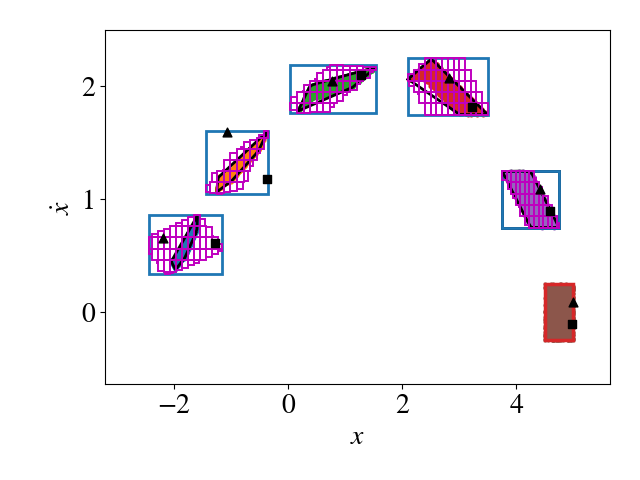}
    \caption{Example of two points that can reach $\mathcal{X}_T$ from $\bar{\mathcal{P}}_{-5}$ without stepping through the true BP set at each step. These points exist because \cref{eqn:hybrid:lp_constr} can only require points to go through $\bar{\mathcal{P}}_t$ and not $\mathcal{P}_t$, leading to non-zero error even with very fine partitioning.}
    \label{fig:lower_bound_explanation}
    \vspace*{-0.18in}
\end{figure}

\subsection{Linearized Ground Robot}
\label{sec:results:lgr}
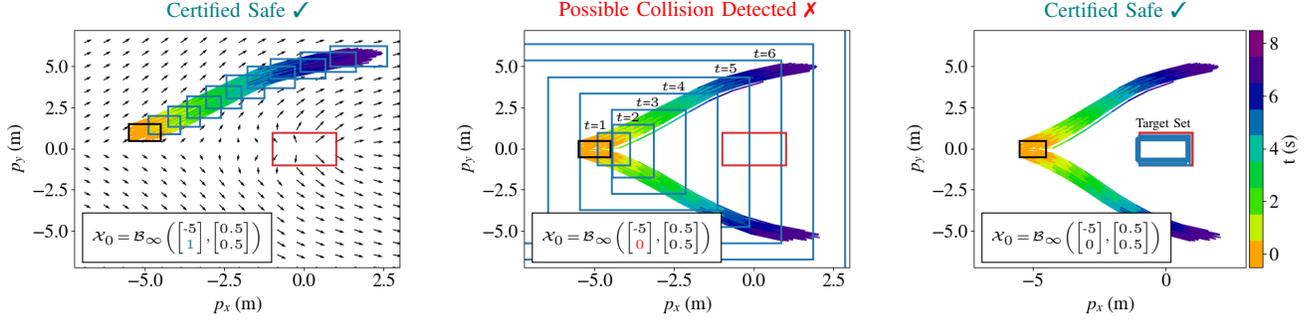
\begin{figure*}[t]
\centering
 \captionsetup[subfigure]{aboveskip=-4pt,belowskip=-4pt}
        \begin{subfigure}[t]{0.325\textwidth}
        \begin{tikzpicture}[fill=white]
            \node[anchor=south west,inner sep=0] (image) at (0,0) {\includegraphics[width=\columnwidth]{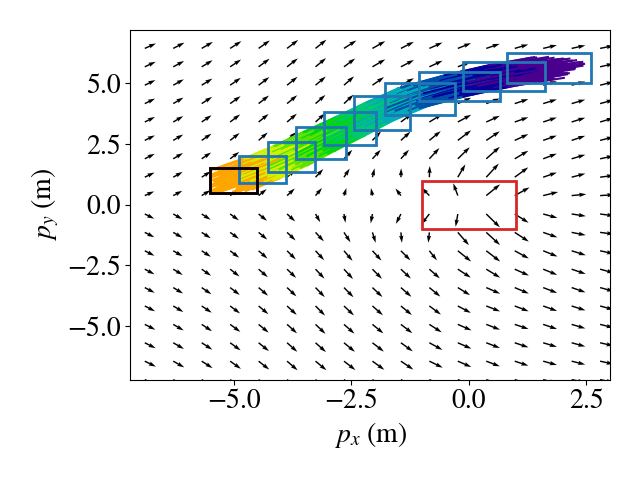}};
            \begin{scope}[x={(image.south east)},y={(image.north west)}]
                \node[draw,fill] at (0.44,0.3) {{\tiny $\mathcal{X}_0\! =\! \mathcal{B}_\infty \! \left( \begin{bmatrix} \text{-}5 \\ {\color{teal}1} \end{bmatrix}\!,\! \begin{bmatrix} 0.5 \\ 0.5 \end{bmatrix} \right)$}};
                % \draw (0.44, 0.5) \node{$t=1$}; 
            \end{scope}
            \begin{scope}[x={(image.south east)},y={(image.north west)}]
                \node[] at (0.58,1) {{\color{teal}\small Certified Safe \ding{51}}};
            \end{scope}
        \end{tikzpicture}
        \caption{{\color{black}Forward reachability for nominal collision avoidance scenario and vector field representation of control input. No intersection of target set (red) and reachable sets (blue) implies that safety can correctly be certified.}}
        \label{fig:forward_reach_nominal}
    \end{subfigure}~
    \begin{subfigure}[t]{0.325\textwidth}
        \begin{tikzpicture}[fill=white]
            \node[anchor=south west,inner sep=0] (image) at (0,0) {\includegraphics[width=\columnwidth]{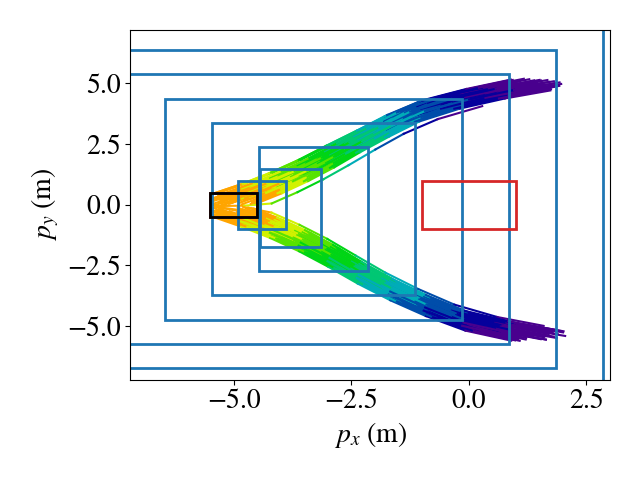}};
            \begin{scope}[x={(image.south east)},y={(image.north west)}]
                \node[draw,fill] at (0.44,0.3) {{\tiny $\mathcal{X}_0\! =\! \mathcal{B}_\infty \! \left( \begin{bmatrix} \text{-}5 \\ {\color{red}0} \end{bmatrix}\!,\! \begin{bmatrix} 0.5 \\ 0.5 \end{bmatrix} \right)$}};
                \node[text width=3cm] at (0.60,0.645) {{\tiny $t\text{=}1$}};
                \node[text width=3cm] at (0.675,0.67) {{\tiny $t\text{=}2$}};
                \node[text width=3cm] at (0.72,0.715) {{\tiny $t\text{=}3$}};
                \node[text width=3cm] at (0.78,0.765) {{\tiny $t\text{=}4$}};
                \node[text width=3cm] at (0.9,0.82) {{\tiny $t\text{=}5$}};
                \node[text width=3cm] at (0.99,0.865) {{\tiny $t\text{=}6$}};
            \end{scope}
            \begin{scope}[x={(image.south east)},y={(image.north west)}]
                \node[] at (0.58,1) {{\color{red}\small Possible Collision Detected \ding{55}}};
            \end{scope}
        \end{tikzpicture}
        \caption{{\color{black}Forward reachability for collision avoidance with bifurcating sets of trajectories. Reachable sets expand to capture both sets of possible trajectories, causing an increase in conservativeness and failure to certify safety.}}
         \label{fig:forward_reach_bifurcating}
    \end{subfigure}~
    \begin{subfigure}[t]{0.325\textwidth}
        \begin{tikzpicture}[fill=white]
            \node[anchor=south west,inner sep=0] (image) at (0,0) {\includegraphics[width=\columnwidth]{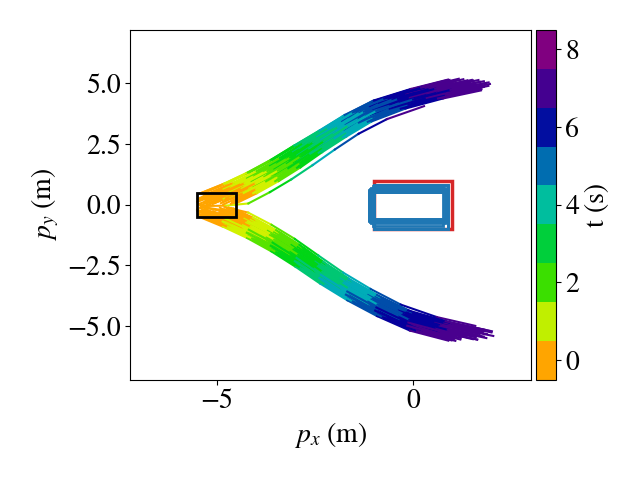}};
            \begin{scope}[x={(image.south east)},y={(image.north west)}]
                \node[draw,fill] at (0.44,0.3) {{\tiny $\mathcal{X}_0\! =\! \mathcal{B}_\infty \! \left( \begin{bmatrix} \text{-}5 \\ 0 \end{bmatrix}\!,\! \begin{bmatrix} 0.5 \\ 0.5 \end{bmatrix} \right)$}};
            \end{scope}
            \begin{scope}[x={(image.south east)},y={(image.north west)}]
                \node[] at (0.53,1) {{\color{teal}\small Certified Safe \ding{51}}};
            \end{scope}
            \begin{scope}[x={(image.south east)},y={(image.north west)}]
                \node[] at (0.64,.65) {{\tiny Target Set}};
            \end{scope}
        \end{tikzpicture}
        \caption{{\color{black}Backward reachability for collision avoidance with bifurcating sets of trajectories. Safety can correctly be certified because none of the BP set estimates (blue) intersect with the initial state set (black).}}
        \label{fig:backward_reach_single_integrator}
    \end{subfigure}
    \caption{Collision avoidance situation that \cite{everett2021reachability} incorrectly labels as dangerous and \hybrid{} correctly certifies as safe.}
    \label{fig:lgr_comparison}
    \vspace*{-0.12in}
\end{figure*}

{\color{black}
Considering the feedback-linearization technique used in \cite{martinez2021formation}, we introduce the common unicycle model, represented as a pair of discrete-time integrators:
\begin{equation}
    \mathbf{x}_{t+1} =
    \underbrace{
    \begin{bmatrix}
    1 & 0 \\
    0 & 1
    \end{bmatrix}}_{\mathbf{A}} \mathbf{x}_t +
    \underbrace{
    \begin{bmatrix}
    1 & 0 \\ 0 & 1
    \end{bmatrix}}_{\mathbf{B}} \mathbf{u}_t,
\end{equation}
with sampling time $t_s=1$s.
Thus, state $\mathbf{x}_t \triangleq [p_x, p_y]$ is the position of a ground robot in the $xy$-plane and the control inputs $\mathbf{u}_t \triangleq [v_x, v_y]$ are the desired $x$ and $y$ velocity components.

To simulate the scenarios shown in \cref{fig:reach_comp}, we trained a 2-layer NN with $[10, 10]$ neurons and ReLU activations with MSE loss to imitate the vector field
\begin{equation}
    \mathbf{u}(\mathbf{x})\text{=}\!\!
    \begin{bmatrix} \!
    \mathrm{clip}(1+\frac{2p_x}{p_x^2+p_y^2},-1,1) \\
    \mathrm{clip}(\frac{p_y}{p_x^2+p_y^2}\!+\!2\mathrm{sign}(p_y)\frac{e^{-\frac{p_x}{2}-2}}{(1+e^{-\frac{p_x}{2}-2})^2},-1,1) \label{eqn:vector_field}\!
    \end{bmatrix},
\end{equation}
where $\mathrm{sign}(\cdot)$ returns the sign of the argument and $\mathrm{clip}(x, a, b)$ constrains the value of $x$ to be within $a$ and $b$.
Eq. \cref{eqn:vector_field} was used to generate $10^5$ state-action pairs by sampling the state space uniformly over $\mathcal{B}_{\infty}([0,0]^\top,[10,10]^\top)$.
The state-action pairs were then used to train the NN for 20 epochs with a batch size of 32.

As shown in \cref{fig:forward_reach_nominal}, the vector field represented by \cref{eqn:vector_field} was designed to produce trajectories that avoid an obstacle at the origin, bounded by the target set $\mathcal{X}_T = \mathcal{B}_{\infty}([0,0]^\top,[1,1]^\top)$ (red).
\cref{fig:forward_reach_nominal} also shows the forward-reachability algorithm Reach-LP \cite{everett2021reachability} in a nominal obstacle-avoidance scenario where the forward reachable set estimates (blue) tightly bound the trajectories ($\mathbf{x}_0 \rightarrow \mathbf{x}_{\tau}: \mathrm{orange} \rightarrow \mathrm{purple}$) sampled from the intial state set $\mathcal{X}_0 = \mathcal{B}_{\infty}([-5,1]^\top,[0.5,0.5]^\top)$.
In this case, $\mathcal{X}_0$ lies strictly above the $x$-axis, resulting in all possible future trajectories going above the obstacle at the origin, thus remaining bunched together and allowing forward reachability to correctly certify safety.

In \cref{fig:forward_reach_bifurcating}, the initial state set is shifted such that it is centered on the $x$-axis, i.e., $\mathcal{X}_0 = \mathcal{B}_{\infty}([-5,0]^\top,[0.5,0.5]^\top)$.
This causes the NN control policy to command the system to go above the obstacle for some states in $\mathcal{X}_0$ while others go below.
The resulting forward reachability analysis shows a drastically different picture than \cref{fig:forward_reach_nominal}, with the forward reachable sets quickly gaining conservativeness as they encompass both sets of possible trajectories.
While other forward reachability analysis tools \cite{sidrane2021overt} may reduce conservativeness in the upper- and lower-bounds of the reachable set estimates, we are not aware of any methods to consider the space between the sets of trajectories, which is the region that intersects with the obstacle and causes safety certification to fail.
Also note that while an alternative strategy could be to partition $\mathcal{X}_0$ such that each element goes in one of the two directions, this would require a perfect split along the learned decision boundary, which would be challenging in practice.

Finally, \cref{fig:backward_reach_single_integrator} shows the same situation as \cref{fig:forward_reach_bifurcating}, but now we use the backward reachability approach described in \cref{sec:approach:linear} as the strategy for safety certification.
Instead of propagating forward from $\mathcal{X}_0$, the states that lead to $\mathcal{X}_T$ are calculated using \hybrid{} with 12 guided partitions to give BP over-approximations (blue).
Since the control policy was designed to force the system away from the obstacle bounded by $\mathcal{X}_T$, the BP over-approximations remain close to $\mathcal{X}_T$ (red) and do not intersect with $\mathcal{X}_0$ (black), thus successfully certifying safety for the situation shown.
Note that while the forward reachable sets in \cref{fig:forward_reach_bifurcating} were calculated in 0.5s compared to 2.0s for the BP over-approximations in \cref{fig:backward_reach_single_integrator}, the result given by \hybrid{} is more useful.

\cref{fig:buggy_NN_demo} shows that, in addition to certifying safe situations, \hybrid{} can also be used to detect failures in a trained control policy.
The policy shown in \cref{fig:buggy_NN_demo} was trained using a modified version of \cref{eqn:vector_field} that causes the policy to lead the system to the origin for some states.
Given this ``bug'' in the policy, the BP over-approximations now expand out to include $\mathcal{X}_0$, implying that if the system is in $\mathcal{X}_0$, it is possible to collide with the obstacle.
}

{\color{red}}

\begin{figure}[t]
    \centering
%    \captionsetup{aboveskip=-12pt,belowskip=-12pt}
    \begin{tikzpicture}[fill=white]
        \node[anchor=south west,inner sep=0] (image) at (0,0) {\includegraphics[width=1\columnwidth]{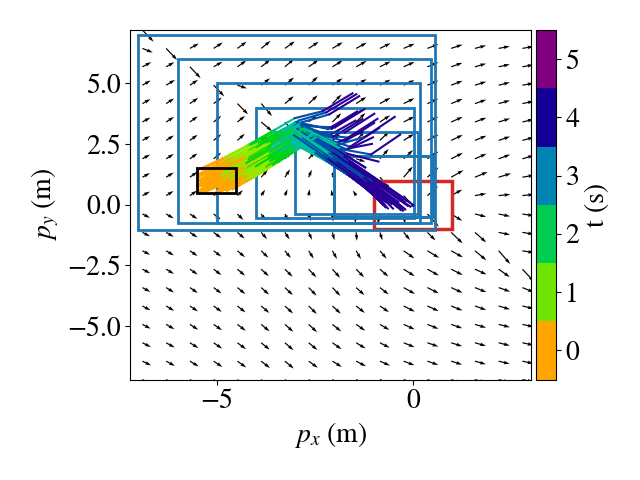}};
        \begin{scope}[x={(image.south east)},y={(image.north west)}]
            \node[draw,fill] at (0.4,0.3) {{\footnotesize $\mathcal{X}_0\! =\! \mathcal{B}_\infty \! \left( \begin{bmatrix} \text{-}5 \\ 1 \end{bmatrix}\!,\! \begin{bmatrix} 0.5 \\ 0.5 \end{bmatrix} \right)$}};
            % \node[text width=3cm] at (0.615,0.645) {{\tiny $t\text{=}1$}};
            % \node[text width=3cm] at (0.675,0.67) {{\tiny $t\text{=}2$}};
            % \node[text width=3cm] at (0.72,0.715) {{\tiny $t\text{=}3$}};
            % \node[text width=3cm] at (0.78,0.765) {{\tiny $t\text{=}4$}};
            % \node[text width=3cm] at (0.9,0.82) {{\tiny $t\text{=}5$}};
            % \node[text width=3cm] at (0.99,0.865) {{\tiny $t\text{=}6$}};
        \end{scope}
        \begin{scope}[x={(image.south east)},y={(image.north west)}]
            \node[] at (0.55,0.97) {{\color{red}\footnotesize Possible Collision Detected \ding{51}}};
        \end{scope}
    \end{tikzpicture}
    \caption{{\color{black}\hybrid{} is does not certify safety for a faulty NN control policy.}}
    \label{fig:buggy_NN_demo}
    \vspace*{-0.18in}
\end{figure}

\subsection{Nonlinear Ground Robot} \label{sec:nonlinear_robot}
We represent a nonlinear version of the ground robot shown in \cref{sec:results:lgr} by changing the dynamics as follows
\begin{equation}
    \begin{aligned}
        p_{x, t+1} = p_{x, t} + v_t\cos(\theta_t) \\
        p_{y, t+1} = p_{y, t} + v_t\sin(\theta_t) 
    \end{aligned}\label{eq:nonlinear_ground_dyn}
\end{equation}
with sampling time $t_s=1$s. The state is still represented by $\mathbf{x}=[p_x,p_y]^\top$, but the control inputs are now $\mathbf{u} = [v, \theta]^\top$ making the dynamics nonlinear in the control inputs due to the multiplicative and trigonometric operations in \cref{eq:nonlinear_ground_dyn}. The velocity and heading control inputs are derived from the vector field described in \cref{eqn:vector_field} and shown in \cref{fig:forward_reach_nominal}. A NN with three hidden layers containing $10$ units each and ReLU activations was trained using $ 1.01 \times 10^6 $ data points generated from the state space region $\mathcal{B}_{\infty}([0,0]^\top,[6,6]^\top)$ to replicate the vector field in \cref{eqn:vector_field}. Data points were more densely sampled near critical regions of the state space (e.g., avoid set). The network was trained for $250$ epochs with a batch size of $512$ and a mean squared error (MSE) loss function.

%We trained a network with three hidden layers containing XXX\stan{Esen fill in} units each to generate these control inputs. \stan{Esen add level of network training details that Nick has (num data points, epochs, batch size, loss fn)}

We can apply forward reachability to assess safety in this scenario using the techniques described by \citet{sidrane2021overt}. However, a forward reachability analysis will behave similarly to the results shown for the linear case in \cref{fig:forward_reach_nominal,fig:forward_reach_bifurcating}, and this scenario, therefore, requires a backward reachability analysis to make a useful safety assessment. \Cref{fig:nonlinear_robot} shows the results of applying \BRNL{} to generate BP sets with $\tau=9$. Since the overapproximated BP sets do not intersect with the initial set, we can verify safety over this time horizon.

\begin{figure}[t]
    \centering
    \begin{tikzpicture}[fill=white]
        \node[anchor=south west,inner sep=0] (image) at (0,0) {\includegraphics[width=1\columnwidth]{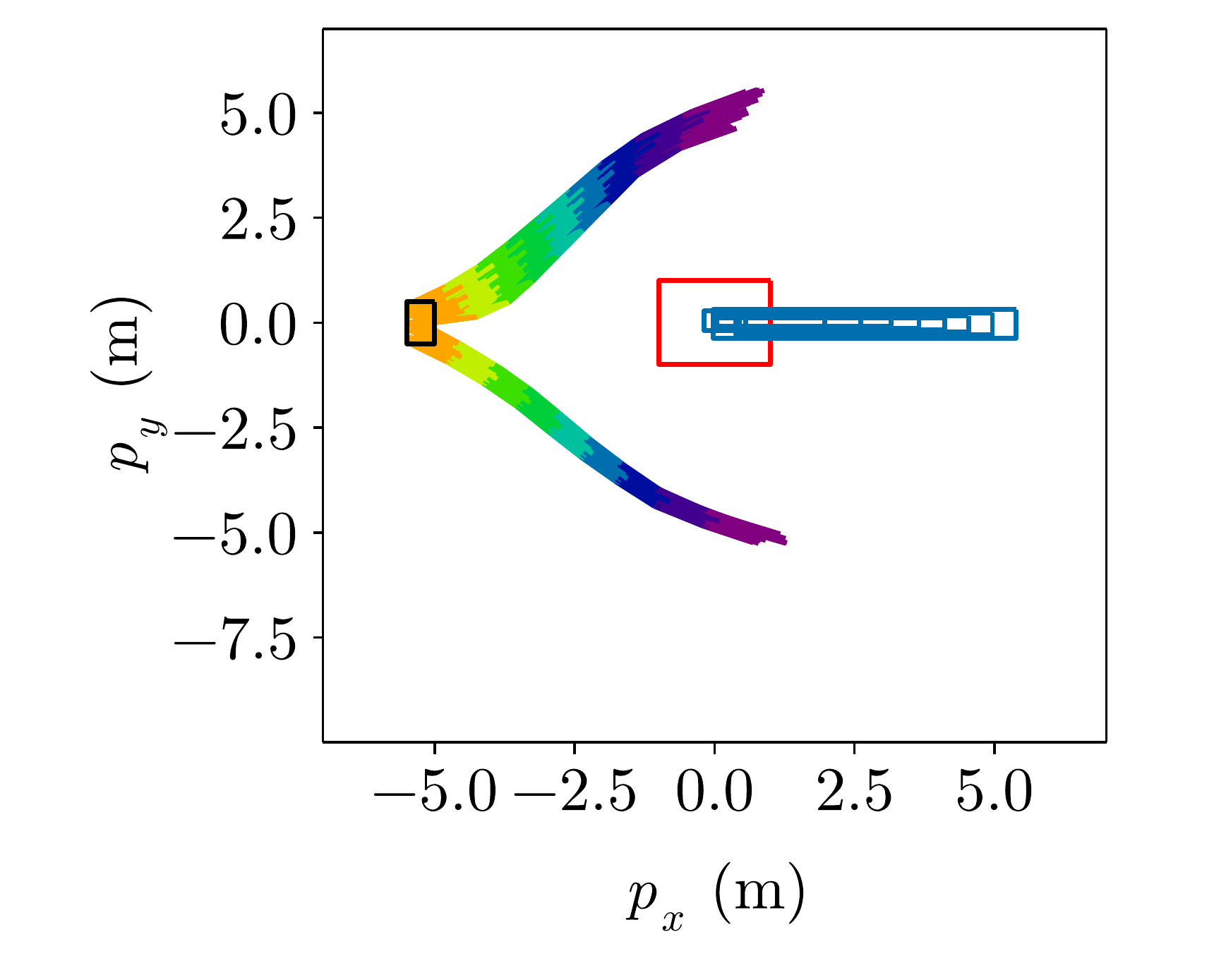}};
        \begin{scope}[x={(image.south east)},y={(image.north west)}]
            \node[draw,fill] at (0.47,0.33) {{\footnotesize $\mathcal{X}_0\! =\! \mathcal{B}_\infty \! \left( \begin{bmatrix} \text{-}5 \\ 1 \end{bmatrix}\!,\! \begin{bmatrix} 0.5 \\ 0.5 \end{bmatrix} \right)$}};
        \end{scope}
        \begin{scope}[x={(image.south east)},y={(image.north west)}]
            \node[] at (0.6,1.01) {{\color{teal}\footnotesize Certified Safe \ding{51}}};
        \end{scope}
        \fill[t0] (7.9, 1.7) rectangle (8.1, 2.256);
        \fill[t1] (7.9, 2.256) rectangle (8.1, 2.811);
        \fill[t2] (7.9, 2.811) rectangle (8.1, 3.367);
        \fill[t3] (7.9, 3.367) rectangle (8.1, 3.922);
        \fill[t4] (7.9, 3.922) rectangle (8.1, 4.478);
        \fill[t5] (7.9, 4.478) rectangle (8.1, 5.033);
        \fill[t6] (7.9, 5.033) rectangle (8.1, 5.589);
        \fill[t7] (7.9, 5.589) rectangle (8.1, 6.144);
        \fill[t8] (7.9, 6.144) rectangle (8.1, 6.7);
        \draw (7.9, 1.7) rectangle (8.1, 6.7);
        \draw (8.1, 1.978) -- (8.18, 1.978);
        \node at (8.3, 1.978) {$0$};
        \draw (8.1, 3.089) -- (8.18, 3.089);
        \node at (8.3, 3.089) {$2$};
        \draw (8.1, 4.2) -- (8.18, 4.2);
        \node at (8.3, 4.2) {$4$};
        \draw (8.1, 5.311) -- (8.18, 5.311);
        \node at (8.3, 5.311) {$6$};
        \draw (8.1, 6.422) -- (8.18, 6.422);
        \node at (8.3, 6.422) {$8$};
        \node at (8.0, 7.0) {t(s)};
    \end{tikzpicture} %\vspace*{-0.4in}
    \caption{Backward reachability analysis for collision avoidance for a nonlinear 2D robot. The controller is certified safe as the backprojection sets (blue) do not intersect with the initial set (black).}
    \label{fig:nonlinear_robot}
    \vspace*{-0.18in}
\end{figure}
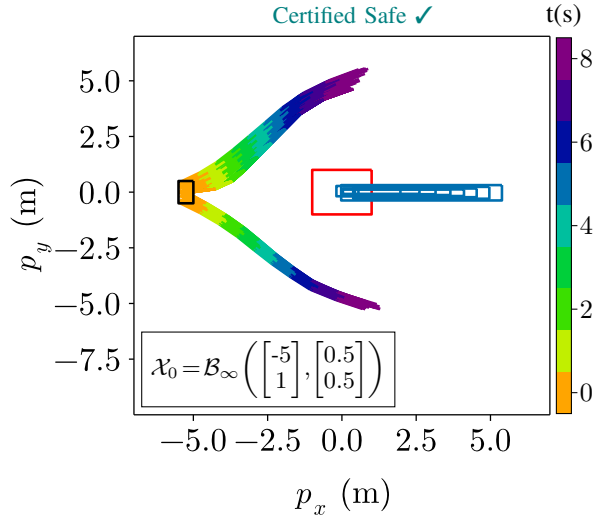

We note that the BP sets for the nonlinear ground robot shown in \cref{fig:nonlinear_robot} do not behave similarly to the BP sets for the linear ground robot shown in \cref{fig:backward_reach_single_integrator}. In particular, while the BP sets for the linear ground robot all remain relatively close to the avoid set, the BP sets for the nonlinear ground robot appear to march to the right in a direction pointed away from the avoid set. This discrepancy provides insight into a potential vulnerability of the system. In particular, the BP sets shown in \cref{fig:nonlinear_robot} indicate that there must be some control outputs near the $x$-axis that cause the robot to move directly to the left and into the avoid set. 

\Cref{fig:nonlinear_samples} confirms this hypothesis by showing samples from the first BP set that reach the avoid set when simulated one step forward in time. 
Upon further analysis, we were able to attribute this behavior to angle wraparound. 
Since we chose to represent $\theta$ using values between $0$ and $2\pi$ radians (where $0$ and $2\pi$ both represent a heading pointing to the right in the $+x$ direction), the vector field in \cref{eqn:vector_field} causes inputs just above the $x$-axis to have an angle near $0$ and inputs just below the $x$-axis to have an angle near $2\pi$. 
Since ReLU networks represent piecewise continuous functions, there must exist a point near the $x$-axis where the control network commands a heading of $\pi$. A heading of $\pi$ causes the robot to point in the $-x$ direction and therefore towards the avoid set, resulting in the rightward marching of the BP sets in \cref{fig:nonlinear_robot}. 
This insight is critical in informing future designs of the system. 
For instance, if we had instead parameterized $\theta$ to use values between $-\pi$ and $\pi$, this vulnerability would occur on the negative portion of the $x$-axis and cause the BP sets to march leftward and intersect with the initial set, ultimately resulting in an unsafe system.

\begin{figure}
    \centering
     \includegraphics[width=0.92\columnwidth]{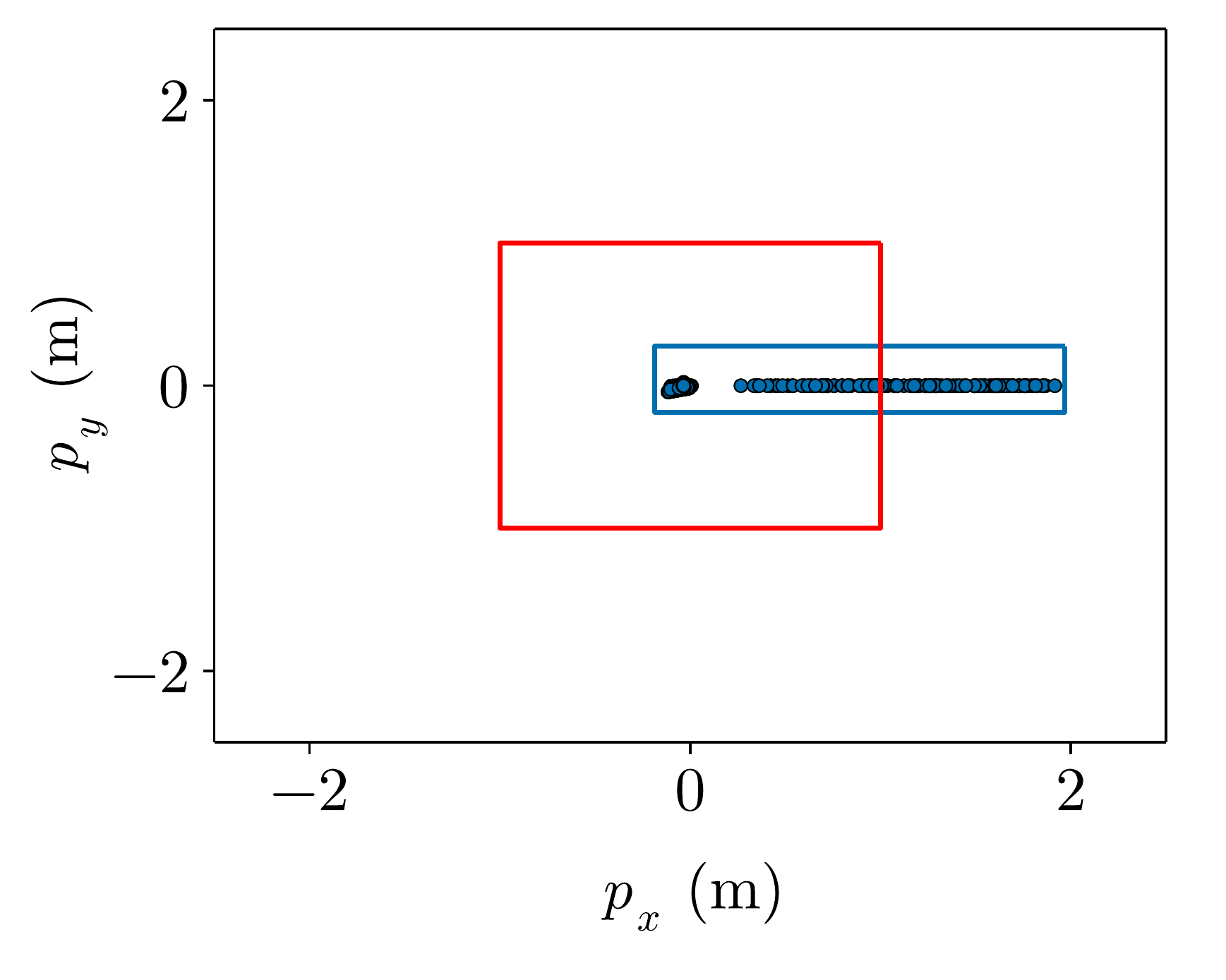}
     \caption{Samples from the one-step BP set (blue) for the nonlinear ground robot that reach the avoid set.}
     \label{fig:nonlinear_samples}
\end{figure}

While \textit{concrete} multi-step backward reachability (as shown in \cref{fig:nonlinear_robot}) provided sufficient accuracy for verifying safety, we also tested a \textit{hybrid-symbolic} approach similar to that described by \citet{sidrane2021overt} for an arbitrary target set. \Cref{fig:symbolic_demo} visualizes the results of performing a symbolic step after every third concrete step. Each set of samples approximates the true BP set at the given time step. In general, the bounds lose tightness over time due to both the accumulation of over-approximation error from the concrete steps and the fact that the true BP sets become less rectangular in shape over time. Nevertheless, the symbolic sets (solid) provide tighter overapproximations than the concrete sets (dashed).

% \begin{figure}
%     \centering
%     \includegraphics[width=\columnwidth]{figures/nonlinear/robot/nstep_traj_v2.pdf}
%     \caption{Backward reachability analysis for collision avoidance for a nonlinear 2D robot. The controller is certified safe as the backprojection sets (blue) do not collide with the initial set (black).}
%     \label{fig:nonlinear_robot}
% \end{figure}

\begin{figure}
    \centering
    \includegraphics[width=0.92\columnwidth]{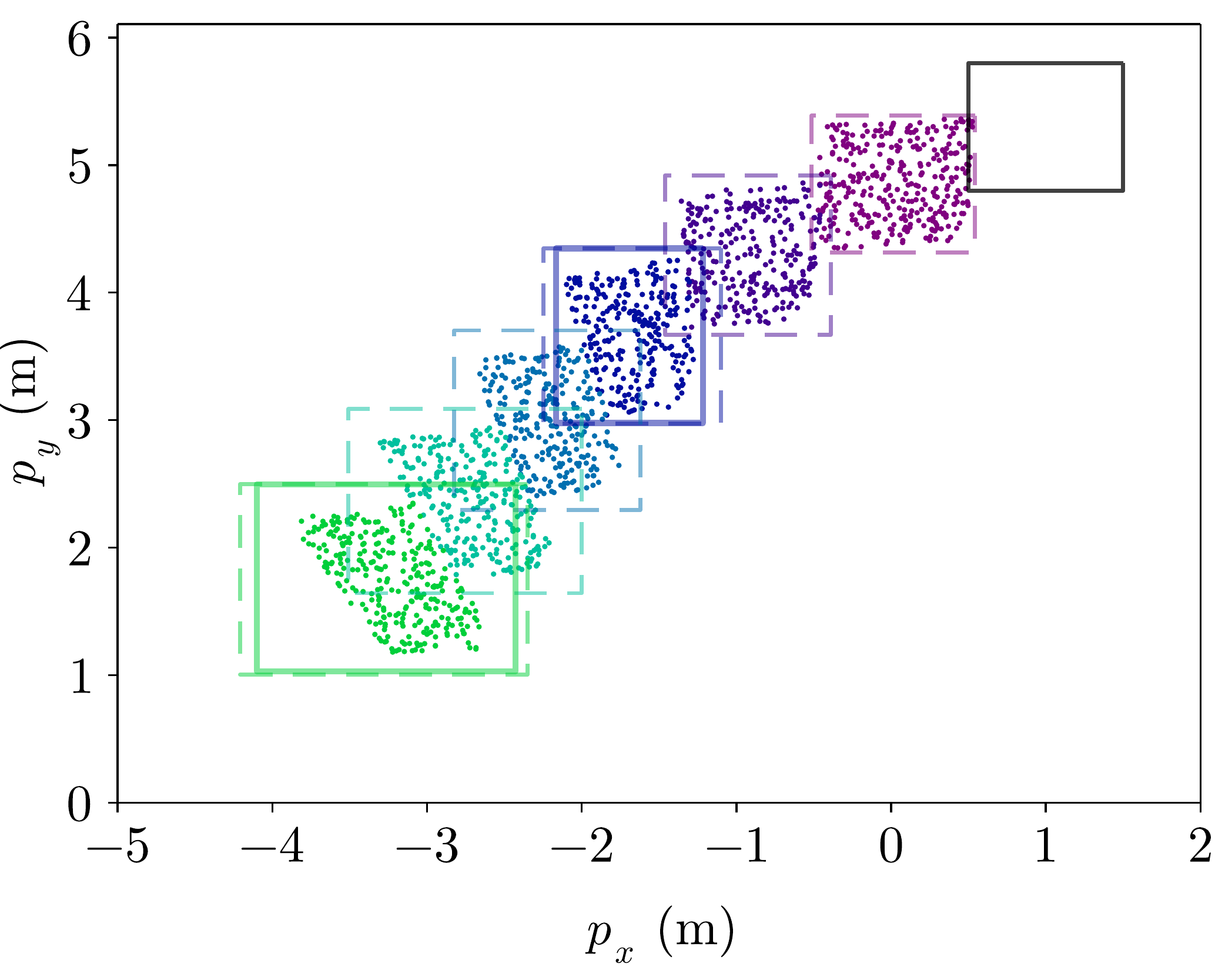}
    \caption{Visualization of concrete (dashed) and symbolic (solid) over-appoximated BP sets corresponding to the target set shown in back. Samples to approximate the true target set were obtained by sampling $50,000$ points uniformly inside the domain and checking for simulations that reach the target set.}
    \label{fig:symbolic_demo}
\end{figure}

\subsection{Ablation Study: Linearized Quadrotor}
To investigate how our algorithm scales with the state space, we next consider a 6-dimensional linearized quadrotor model. By representing the quadrotor as three double integrators corresponding to the vehicle's $x-y-z$ coordinates, we obtain the model
{\small
\begin{equation}
    \mathbf{x}_{t+1} =
    \underbrace{
    \begin{bmatrix}
    1 & 0 & 0 & 1 & 0 & 0 \\
    0 & 1 & 0 & 0 & 1 & 0 \\
    0 & 0 & 1 & 0 & 0 & 1 \\
    0 & 0 & 0 & 1 & 0 & 0 \\
    0 & 0 & 0 & 0 & 1 & 0 \\
    0 & 0 & 0 & 0 & 0 & 1 \\
    \end{bmatrix}}_{\mathbf{A}} \mathbf{x}_t +
    \underbrace{
    \begin{bmatrix}
    0.5 & 0 & 0 \\
    0 & 0.5 & 0 \\
    0 & 0 & 0.5 \\
    1 & 0 & 0 \\
    0 & 1 & 0 \\
    0 & 0 & 1 \\
    \end{bmatrix}}_{\mathbf{B}} \mathbf{u}_t.
\end{equation}
}
% with
% \begin{align*}
%     \mathcal{X} = \left\{ 
%     \begin{bmatrix}
%         -\infty, & \infty \\
%         -\infty, & \infty \\
%         -\infty, & \infty \\
%         -1, & 1 \\
%         -1, & 1 \\
%         -1, & 1 \\
%     \end{bmatrix} 
%     \right\} & \quad
%     \mathcal{U} = \left\{ 
%     \begin{bmatrix}
%         -4, & 4 \\
%         -4, & 4 \\
%         -4, & 4 \\
%     \end{bmatrix} 
%     \right\}
% \end{align*}
In a similar way to \cref{sec:results:lgr}, we trained a control policy with $[20, 20]$ ReLU-activated neurons by specifying a vector field that avoids an obstacle placed at $\mathcal{X}_T = \mathcal{B}_{\infty}([0,0, 2.5, 0,0,0]^\top,[1,1,1,1,1,1]^\top)$. The vector field is represented by the equations
\begin{equation}
    \mathbf{u} = \begin{cases}
    4\mathrm{sign}(\mathbf{p}) \quad & \|\mathbf{p}\|_{\infty} < 2.25 \\
    
    \begin{bmatrix}
        0 \\
        \mathrm{clip}(\frac{4}{p_y}-\frac{p_y}{4}-3v_y, -2, 2) \\
        0 \\
    \end{bmatrix} \quad & \|\mathbf{p}\|_{\infty} \geq 2.25.
    \end{cases} \label{eqn:linearized_quad_policy}
\end{equation}

\cref{fig:6d_linear_quad} shows simulated trajectories using the controller trained with \cref{eqn:linearized_quad_policy} propagated from an initial condition at $\mathcal{X}_0 = \mathcal{B}_{\infty}([-5,0,2.5,0.97,0,0]^\top,$ $[0.25,0.25,0.25,0.02,0.01,0.01]^\top)$ with $dt=1$ over a time horizon $\tau = 6$.
Like \cref{fig:backward_reach_single_integrator}, the BP set estimates (blue), obtained with \cref{alg:hybridBackprojection} using $r=750$ do not intersect with the initial state set, thus proving safety in this scenario. This is confirmed by observing that the simulated trajectories do not intersect with the obstacle (red).
Note that while BP over-approximations are calculated through $\tau = 6$, the fact that the first BP over-approximation is fully contained within the target set implies that all subsequent BP over-approximations should also fall within the target set, allowing us to stop after the calculation of the first BP over-approximation.
The idea that $\bar{\mathcal{P}}_{t}\! \subseteq \! \mathcal{P}_{t+1}$ leads to a form of invariance is stated formally and proven in \cref{sec:proof_invariance}.

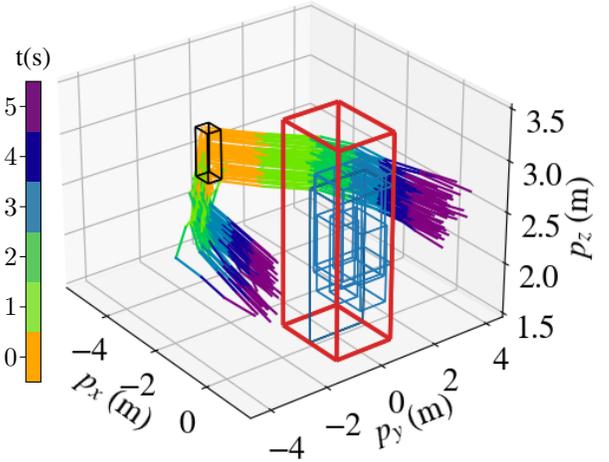
\begin{figure}[t]
    \centering
%    \captionsetup{aboveskip=-12pt,belowskip=-12pt}
    % \includegraphics[width=1\columnwidth]{figures/extension/discrete_quad2.png}
    \begin{tikzpicture}[fill=white]
        \node[anchor=south west,inner sep=0] (image) at (0,0) {\includegraphics[width=1\columnwidth]{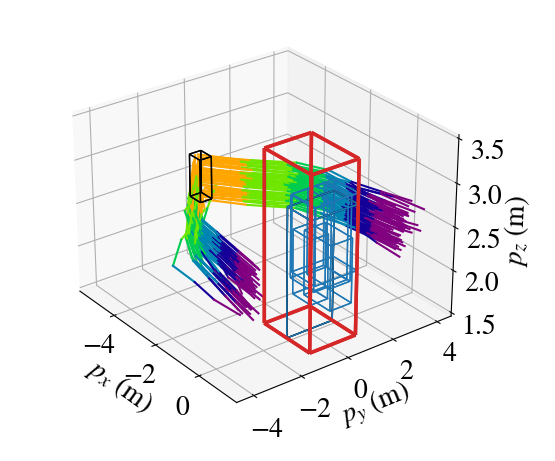}};
        % \begin{scope}[x={(image.south east)},y={(image.north west)}]
        %     \node[] at (0.6,1.01) {{\color{teal}\footnotesize Certified Safe \ding{51}}};
        % \end{scope}
        
        \fill[t02] (8.2+\xshift, 0.7+\yshift) rectangle (8.4+\xshift, 1.367+\yshift);
        \fill[t12] (8.2+\xshift, 1.367+\yshift) rectangle (8.4+\xshift, 2.033+\yshift);
        \fill[t22] (8.2+\xshift, 2.033+\yshift) rectangle (8.4+\xshift, 2.700+\yshift);
        \fill[t32] (8.2+\xshift, 2.700+\yshift) rectangle (8.4+\xshift, 3.367+\yshift);
        \fill[t42] (8.2+\xshift, 3.367+\yshift) rectangle (8.4+\xshift, 4.033+\yshift);
        \fill[t52] (8.2+\xshift, 4.033+\yshift) rectangle (8.4+\xshift, 4.700+\yshift);
        % \draw (8.2+\xshift, 0.7+\yshift) rectangle (8.4+\xshift, 4.7+\yshift);
        % \draw (8.4+\xshift, 1.033+\yshift) -- (8.48+\xshift, 1.033+\yshift);
        % \node at (8.6+\xshift, 1.033+\yshift) {$0$};
        % \draw (8.4+\xshift, 1.700+\yshift) -- (8.48+\xshift, 1.700+\yshift);
        % \node at (8.6+\xshift, 1.700+\yshift) {$1$};
        % \draw (8.4+\xshift, 2.367+\yshift) -- (8.48+\xshift, 2.367+\yshift);
        % \node at (8.6+\xshift, 2.367+\yshift) {$2$};
        % \draw (8.4+\xshift, 3.033+\yshift) -- (8.48+\xshift, 3.033+\yshift);
        % \node at (8.6+\xshift, 3.033+\yshift) {$3$};
        % \draw (8.4+\xshift, 3.700+\yshift) -- (8.48+\xshift, 3.700+\yshift);
        % \node at (8.6+\xshift, 3.700+\yshift) {$4$};
        % \draw (8.4+\xshift, 4.367+\yshift) -- (8.48+\xshift, 4.367+\yshift);
        % \node at (8.6+\xshift, 4.367+\yshift) {$5$};
        % \node at (8.3+\xshift, 5.0+\yshift) {t(s)};
        \draw (8.2+\xshift, 0.7+\yshift) rectangle (8.4+\xshift, 4.7+\yshift);
        \draw (8.2+\xshift, 1.033+\yshift) -- (8.12+\xshift, 1.033+\yshift);
        \node at (8.0+\xshift, 1.033+\yshift) {$0$};
        \draw (8.2+\xshift, 1.700+\yshift) -- (8.12+\xshift, 1.700+\yshift);
        \node at (8.0+\xshift, 1.700+\yshift) {$1$};
        \draw (8.2+\xshift, 2.367+\yshift) -- (8.12+\xshift, 2.367+\yshift);
        \node at (8.0+\xshift, 2.367+\yshift) {$2$};
        \draw (8.2+\xshift, 3.033+\yshift) -- (8.12+\xshift, 3.033+\yshift);
        \node at (8.0+\xshift, 3.033+\yshift) {$3$};
        \draw (8.2+\xshift, 3.700+\yshift) -- (8.12+\xshift, 3.700+\yshift);
        \node at (8.0+\xshift, 3.700+\yshift) {$4$};
        \draw (8.2+\xshift, 4.367+\yshift) -- (8.12+\xshift, 4.367+\yshift);
        \node at (8.0+\xshift, 4.367+\yshift) {$5$};
        \node at (8.3+\xshift, 5.0+\yshift) {t(s)};
    \end{tikzpicture} %\vspace*{-0.4in}
    \caption{Verification of 6-dimensional quadrotor system during obstacle avoidance maneuver.}
    \label{fig:6d_linear_quad}
    %\vspace*{-0.18in}
\end{figure}

To help quantify the improvements in computation time made to our approach when compared with \cite{rober2022backward}, we replicated the scenario showed in \cref{fig:6d_linear_quad} with several different configurations of \hybrid{}. \cref{tab:linear_quad_ablation} shows the results of our comparison. Note that rows using a uniform partitioning strategy split each state into 5 partitions (resulting in $5^6$ total elements), which is the smallest amount capable of correctly verifying safety.
The first row represents \cite{rober2022backward} with the exception that it uses \hybrid{} rather than \basic{} or \nstep{}.
The second row makes use of disciplined parametric programming (DPP) \cite{diamond2016cvxpy}, which allows {\tt cvxpy} to solve many LPs more efficiently.
Finally, the third row includes the SkipLP strategy described in \cref{sec:approach:partition} and the last row uses all components together with the guided partitioning strategy described in \cref{alg:guided_partition}.

\begin{table}[t]
\centering
\caption{Ablation study showing the computational cost of calculating \cref{fig:6d_linear_quad} with different combinations of DPP, SkipLP, and guided partitioning.}
\begin{tabular}{@{}lr@{}}
\toprule
 Configuration   & Runtime [s] \\
\midrule
 Uniform \cite{rober2022backward} & $3796 $ \\
 Uniform + DPP & $2064 $ \\
 Uniform + DPP + SkipLP & $1365 $ \\
 Guided + DPP + SkipLP   & $28 $ \\
\bottomrule
\end{tabular}
\label{tab:linear_quad_ablation}
\vspace*{-0.12in}
\end{table}

\section{Conclusion}
In this paper, we considered the problem of backward reachability for NFLs.
We presented an algorithm that combines the best elements of the two algorithms proposed by \cite{rober2022backward} to calculate over-approximations of BP sets over a given time horizon for systems with linear dynamics. We also introduced a new algorithm to compute over-approximations of the BP sets for nonlinear dynamical systems.
We highlighted the advantages of backward reachability on both linear and nonlinear versions of a 2D collision avoidance system.
We also demonstrated reduced computation time by introducing a guided partitioning strategy that more intelligently divides up the input space to the NN relaxation used to find the BP sets.
These improvements were used to verify a 6D linearized quadrotor system over 100 times faster than the previous method.

\balance
% \bibliographystyle{IEEEtran}
% \bibliography{refs}
\printbibliography

\newpage
\appendix
\subsection{Proof of \texorpdfstring{\cref{thm:hybridbackprojection}}{Lemma 4.1}}\label{sec:proof_hybridbackprojection}

\begingroup
\allowdisplaybreaks

\thmhybridbackprojection*
\begin{proof}
Consider the calculation of the first BP set over-approximation $\bar{\mathcal{P}}_{-1}$. Using \cref{eqn:backreachable:lp_obj}, we first calculate $\bar{\mathcal{R}}_{-1}$, which is guaranteed to contain all the states for which the system can arrive in $\mathcal{X}_T$ given $\mathbf{u} \in \mathcal{U}$ and $\mathbf{x} \in \mathcal{X}$. Notice that $\bar{\mathcal{P}}_{-1} \subseteq \bar{\mathcal{R}}_{-1}$, which allows us to apply \cref{thm:crown_particular_x} over $\bar{\mathcal{R}}_{-1}$ to obtain bounds $\pi^L_{-1}(\mathbf{x})$ and $\pi^U_{-1}(\mathbf{x})$ that are valid over $\mathcal{P}_{-1}$. We then use \cref{eqn:hybrid:lp_constr} to include the constraint that $\pi^L_{-1}(\x) \leq \mathbf{u}_{-1} \leq \pi^U_{-1}(\x)$, which is guaranteed by \cref{thm:crown_particular_x} to satisfy $\pi^L_{-1}(\x_{-1}) \leq \pi(\mathbf{x}) \leq \pi^U_{-1}(\x_{-1})$. From this we can conclude that
\begin{align}
    \bar{\mathcal{P}}_{-1}(\mathcal{X}_T) = &\ \{\x\ |\ \ubar{\x}_{-1} \leq \x \leq \bar{\x}_{-1} \}
    \label{eqn:BPproof:rect_bound} \\ 
    \begin{split}
         \supseteq & \ \{\x\ |\ f(\x, \mathbf{u}) \in \mathcal{X}_T,\ \\ & \quad \quad \quad \quad \pi^L_{-1}(\x) \leq \mathbf{u} \leq \pi^U_{-1}(\x)\},
    \end{split} \label{eqn:BPproof:true_relaxed_BP} \\
    \supseteq &\ \{ \mathbf{x}\ \lvert\ f(\mathbf{x}, \pi(\mathbf{x})) \in \mathcal{X}_T \} = \mathcal{P}_{-1}(\mathcal{X}_T),
    \label{eqn:BPproof:true_BP}
\end{align}
thus proving \cref{thm:hybridbackprojection} for the first time step. This gives us the result that all states that reach $\mathcal{X}_T$ must get there by applying a control value satisfying $\pi^L_{-1}(\x) \leq \mathbf{u}_{-1} \leq \pi^U_{-1}(\x)$ from somewhere in $\bar{\mathcal{P}}_{-1}$.

Now if we consider \cref{eqn:hybrid:lp_obj} for $t=-2$, i.e.,
\begin{align}
    \bar{\mathcal{P}}_{-2}(\mathcal{X}_T) = &\ \{\x_{-2}\ |\ \ubar{\x}_{-2} \leq \x_{-2} \leq \bar{\x}_{-2} \}
    \label{eqn:BPproof:t2rect_bound} \\ 
    \begin{split}
         \supseteq 
         & \ \{\x_{-2}\ |\ \mathbf{x}_{-1} = f(\x_{-2}, \mathbf{u}_{-2}) \in \bar{\mathcal{P}}_{-1},\ \\ 
         & \quad \quad \quad \ f(\mathbf{x}_{-1}, \mathbf{u_{-1}}) \in \mathcal{X}_T,\ \\ 
         & \quad \quad \quad \ \pi^L_{-2}(\x) \leq \mathbf{u}_{-2} \leq \pi^U_{-2}(\x), \\
         & \quad \quad \quad \ \pi^L_{-1}(\x) \leq \mathbf{u}_{-1} \leq \pi^U_{-1}(\x)\},
    \end{split} \label{eqn:BPproof:t2true_relaxed_BP} \\
    \begin{split}
    \supseteq 
    & \ \{ \mathbf{x}_{-2}\ \lvert\ \mathbf{x}_{-1} =  f(\mathbf{x}_{-2}, \pi(\mathbf{x}_{-2})) \in \mathcal{P}_{-1}, \\ & f(\mathbf{x}_{-1}, \pi(\mathbf{x}_{-1})) \in \mathcal{X}_T, \} = \mathcal{P}_{-2}(\mathcal{X}_T),
    \end{split}
    \label{eqn:BPproof:t2true_BP}
\end{align}
it becomes clear that because $\mathcal{P}_{-1}(\mathcal{X}_T) \subseteq \bar{\mathcal{P}}_{-1}(\mathcal{X}_T)$, and with bounds on the control limits from \cref{thm:crown_particular_x}, we can conclude that $\mathcal{P}_{-2}(\mathcal{X}_T) \subseteq \bar{\mathcal{P}}_{-2}(\mathcal{X}_T)$. From here, the same argument can be applied an arbitrary number of time steps backward, allowing us to conclude
\begin{equation}
    \mathcal{P}_{t}(\mathcal{X}_T) \subseteq \bar{\mathcal{P}}_{t}(\mathcal{X}_T) \quad \forall t \in \mathcal{T}.
\end{equation}
\end{proof}
\endgroup

{\color{black}
\subsection{Backward Invariance} \label{sec:proof_invariance}

\begin{lemma} \label{lem:invariance}
If there exists a BP over-approximation $\bar{\mathcal{P}}_{t_2}$ and a concrete set $\mathcal{P}_{t_1}$ at times $t_1, t_2$ with $t_2 < t_1 \leq 0$, such that $\bar{\mathcal{P}}_{t_2} \subseteq \mathcal{P}_{t_1}$, then
\begin{equation}
    \mathcal{I} = \mathcal{P}_{t_1} \cup \bigcup_{t=t_2}^{t_1-1} \bar{\mathcal{P}}_t
\end{equation}
contains all BP sets for $t \leq t_1$, i.e., $\bar{\mathcal{P}}_{t} \subseteq \mathcal{I}\ \forall t \leq t_1$. 
\end{lemma}

\begin{proof}
This is a proof by induction where we wish to show that for any $t \leq t_2$, $\mathcal{P}_{t} \subseteq \mathcal{I}$.
To prove the base case, we simply have to show that $\mathcal{P}_{t_2} \subseteq \mathcal{I}$. Given that $\mathcal{P}_{t_2} \subseteq \bar{\mathcal{P}}_{t_2}$ and that $\mathcal{I} = \mathcal{P}_{t_1} \cup \dots \cup \bar{\mathcal{P}}_{t_2}$, it must be that $\mathcal{P}_{t_2} \subseteq \mathcal{I}$.

To show the inductive step, we need to show that for any $t < t_2$, $\mathcal{P}_{t} \subseteq \mathcal{I}$.
First note that given $\bar{\mathcal{P}}_{t_2} \subseteq \mathcal{P}_{t_1}$ and that $\mathcal{P}_{t_2} \subseteq \bar{\mathcal{P}}_{t_2}$, 
we know that $\mathcal{P}_{t_2} \subseteq \mathcal{P}_{t_1}$. It then follows that $\mathcal{P}_{t_2-1} \subseteq \mathcal{P}_{t_1-1}$ because if $\mathcal{P}_{t_2} \subseteq \mathcal{P}_{t_1}$, any state that leads to $\mathcal{P}_{t_2}$ must also lead to $\mathcal{P}_{t_1}$.
This argument can be applied iteratively to obtain
\begin{equation*}
    \begin{split}
        \mathcal{P}_{t_2-1} & \subseteq \mathcal{P}_{t_1-1} \\
        \mathcal{P}_{t_2-2} & \subseteq \mathcal{P}_{t_1-2} \\
        &\ \vdots \\
        \mathcal{P}_{t_2-\Delta t} & \subseteq \mathcal{P}_{t_1-\Delta t}.
    \end{split}
\end{equation*}
By selecting $\Delta t \triangleq t_1-t_2$, we can make the substitution $\mathcal{P}_{t_1-\Delta t} =  \mathcal{P}_{t_2}$ to get
\begin{equation*}
    \mathcal{P}_{t_2 - \Delta t} \subseteq \mathcal{P}_{t_2} \subseteq \mathcal{P}_{t_1}.
\end{equation*}
It follows that $\mathcal{P}_{t} \subseteq \mathcal{P}_{t +  \Delta t}\ \forall t \leq t_2$.
If we assume that $\bar{\mathcal{P}}_{i} \subseteq \mathcal{I}\ \forall i \in \{t+1, \dots t_2\}$, then we get $\mathcal{P}_{t} \subseteq \mathcal{P}_{t +  \Delta t} \subseteq \bar{\mathcal{P}}_{t +  \Delta t} \subseteq \mathcal{I}$, thus proving the inductive step.
\end{proof}

\begin{corollary}
If there exists a BP overapproximation at time $t=-1$ such that $\bar{\mathcal{P}}_{-1} \subseteq \mathcal{X}_{T}$, then $\mathcal{X}_{T}$ 
contains all BP sets with $t \leq t_1$, i.e., $\mathcal{X}_{T} \subseteq \mathcal{I}\ \forall t \leq t_1$, thus rendering $\mathcal{X}_{T}$ an invariant set. 
\end{corollary}}

\end{document}